\documentclass[11pt,letterpaper]{article}
\usepackage[latin1]{inputenc}
\usepackage[OT4]{fontenc}
\usepackage{amsmath}
\usepackage{amsthm}
\usepackage{amsfonts}
\usepackage{amssymb}
\usepackage{thmtools}
\usepackage{mathtools}
\usepackage{graphicx}
\usepackage{thm-restate}

\usepackage{tikz}
\usepackage{multirow}
\usetikzlibrary{calc,shapes,patterns}
\definecolor{mygreen}{RGB}{00,255,180}
\newlength{\hatchspread}
\newlength{\hatchthickness}
\newlength{\hatchshift}
\newcommand{\hatchcolor}{}
\tikzset{hatchspread/.code={\setlength{\hatchspread}{#1}},
	hatchthickness/.code={\setlength{\hatchthickness}{#1}},
	hatchshift/.code={\setlength{\hatchshift}{#1}},%
	hatchcolor/.code={\renewcommand{\hatchcolor}{#1}}}
\tikzset{hatchspread=3pt,
	hatchthickness=0.4pt,
	hatchshift=0pt,%
	hatchcolor=black}
\pgfdeclarepatternformonly[\hatchspread,\hatchthickness,\hatchshift,\hatchcolor]%
{custom north west lines}%
{\pgfqpoint{\dimexpr-2\hatchthickness}{\dimexpr-2\hatchthickness}}%
{\pgfqpoint{\dimexpr\hatchspread+2\hatchthickness}{\dimexpr\hatchspread+2\hatchthickness}}%
{\pgfqpoint{\dimexpr\hatchspread}{\dimexpr\hatchspread}}%
{%
	\pgfsetlinewidth{\hatchthickness}
	\pgfpathmoveto{\pgfqpoint{0pt}{\dimexpr\hatchspread+\hatchshift}}
	\pgfpathlineto{\pgfqpoint{\dimexpr\hatchspread+0.15pt+\hatchshift}{-0.15pt}}
	\ifdim \hatchshift > 0pt
	\pgfpathmoveto{\pgfqpoint{0pt}{\hatchshift}}
	\pgfpathlineto{\pgfqpoint{\dimexpr0.15pt+\hatchshift}{-0.15pt}}
	\fi
	\pgfsetstrokecolor{\hatchcolor}
	\pgfusepath{stroke}
}
\pgfdeclarepatternformonly[\hatchspread,\hatchthickness,\hatchshift,\hatchcolor]%
{custom north east lines}%
{\pgfqpoint{\dimexpr-2\hatchthickness}{\dimexpr-2\hatchthickness}}%
{\pgfqpoint{\dimexpr\hatchspread+2\hatchthickness}{\dimexpr\hatchspread+2\hatchthickness}}%
{\pgfqpoint{\dimexpr\hatchspread}{\dimexpr\hatchspread}}%
{%
	\pgfsetlinewidth{\hatchthickness}
	\pgfpathmoveto{\pgfqpoint{\dimexpr\hatchshift-0.15pt}{-0.15pt}}
	\pgfpathlineto{\pgfqpoint{\dimexpr\hatchspread+0.15pt}{\dimexpr\hatchspread-\hatchshift+0.15pt}}
	\ifdim \hatchshift > 0pt
	\pgfpathmoveto{\pgfqpoint{-0.15pt}{\dimexpr\hatchspread-\hatchshift-0.15pt}}
	\pgfpathlineto{\pgfqpoint{\dimexpr\hatchshift+0.15pt}{\dimexpr\hatchspread+0.15pt}}
	\fi
	\pgfsetstrokecolor{\hatchcolor}
	\pgfusepath{stroke}
}
\usepackage[textwidth=3.5cm,textsize=scriptsize]{todonotes}

\newcommand{\nr}{\operatorname{nr}}
\newcommand{\BC}{\operatorname{BC}}
\newcommand{\CBC}{\overline{\operatorname{BC}}}
\newcommand{\CB}{\overline{\operatorname{B}}}
\newcommand{\boundary}{\partial}

\newcommand{\connected}{\operatorname{connected}}
\newcommand{\cid}{\operatorname{cid}}
\newcommand{\meet}{\operatorname{meet}}
\newcommand{\poly}{\operatorname{poly}}
\newcommand{\polylog}{\operatorname{polylog}}
\newcommand{\splitnode}{\operatorname{split}}
\newcommand{\contract}{\operatorname{contract}}
\newcommand{\deleteedge}{\operatorname{delete}}
\newcommand{\parent}{\operatorname{parent}}
\newcommand{\children}{\operatorname{children}}
\newcommand{\firstchild}{\operatorname{first-child}}
\newcommand{\nextsibling}{\operatorname{next-sibling}}
\newcommand{\nca}{\operatorname{nca}}
\newcommand{\ca}{\operatorname{ca}}
\newcommand{\firstonpath}{\operatorname{first-on-path}}
\newcommand{\apex}{\operatorname{apex}}
\newcommand{\LOinsert}{\operatorname{insert}}
\newcommand{\LOdelete}{\operatorname{delete}}
\newcommand{\LOorder}{\operatorname{order}}
\newcommand{\LOsucc}{\operatorname{succ}}
\newcommand{\LOpred}{\operatorname{pred}}
\newcommand{\CAaddleaf}{\operatorname{add-leaf}}
\newcommand{\CAdeleteleaf}{\operatorname{delete-leaf}}
\newcommand{\rank}{\operatorname{rank}}
\newcommand{\nil}{\mathbf{nil}}
\newcommand{\false}{\mathbf{false}}
\newcommand{\true}{\mathbf{true}}

\DeclarePairedDelimiter{\abs}{\lvert}{\rvert}
\DeclarePairedDelimiter{\ceil}{\lceil}{\rceil}
\DeclarePairedDelimiter{\floor}{\lfloor}{\rfloor}

\DeclarePairedDelimiter{\set}{\lbrace}{\rbrace}
\newcommand{\cond}{\mathrel{}\middle\vert\mathrel{}}

\DeclarePairedDelimiter{\Paren}{(}{)}

\declaretheorem[name={Theorem},style=plain]{theorem}
\declaretheorem[name={Lemma},sibling=theorem]{lemma}
\declaretheorem[name={Corollary},sibling=theorem]{corollary}
\declaretheorem[name={Observation},sibling=theorem]{observation}
\declaretheorem[name={Definition},sibling=theorem,style=definition]{definition}

\usepackage{authblk}
\usepackage{hyperref}
\title{Good $r$-divisions Imply Optimal Amortized Decremental Biconnectivity}
\author[1]{Jacob Holm}
\author[2]{Eva Rotenberg}
\affil[1]{University of Copenhagen \hspace{1em}{\small \href{mailto:jaho@di.ku.dk}{jaho@di.ku.dk}}}
\affil[2]{Technical University of Denmark \hspace{1em}{\small \href{mailto:eva@rotenberg.dk}{erot@dtu.dk}}}
\date{}
\begin{document}
\thispagestyle{empty}
\maketitle

\begin{abstract}
    We present a data structure that, given a graph $G$ of $n$ vertices and $m$ edges, and a suitable pair of nested $r$-divisions of $G$, preprocesses $G$ in $O(m+n)$ time and handles any series of edge-deletions in $O(m)$ total time while answering queries to pairwise biconnectivity in worst-case $O(1)$ time. In case the vertices are not biconnected, the data structure can return a cutvertex separating them in worst-case $O(1)$ time.

    As an immediate consequence, this gives optimal amortized decremental  biconnectivity, 2-edge connectivity, and connectivity for large classes of graphs, including planar graphs and other minor free graphs. 
\end{abstract}

\thispagestyle{empty}
\newpage
\setcounter{page}{1} 

\section{Introduction}

\emph{Dynamic graph problems} concern maintaining information about a graph, as it undergoes changes. In this paper, the changes we allow are deletions of edges or vertices by an adaptive adversary. The information we maintain is a representation that reflects biconnectivity of vertices, that is, whether they are connected after the removal of any vertex of the graph.

A static (non-changing) graph may in $O(n+m)$ time be pre-processed to answer biconnectivity queries in worst-case $O(1)$ time. This is done by finding the \emph{blocks}, i.e. the biconnected components. We show, for a large class of graphs including minor free graphs, that in the same asymptotic total time, we can handle any sequence of edge- and vertex deletions, while still answering biconnectivity queries, 2-edge connectivity queries, and connectivity queries, in worst-case $O(1)$ time. 

If a pair of vertices are not biconnected, then there exists a certificate for this %
in form of a cutvertex separating them. A natural question, if a pair of vertices are not biconnected, is thus to ask for such a certificate. There may be many cutvertices separating a pair of vertices, so an even more advanced and desired functionality is the ability to point to the one nearest to a specified one of them. Again, for a large class of graphs, our running time for a \emph{decremental} graph matches the state of the art for \emph{non-changing graphs}, by revealing the nearest cutvertex in $O(1)$ worst-case time, while processing deletions in $O(n+m)$ total time. 

When edges and vertices are both deleted and inserted, there are non-trivial lower bounds~\cite{patrascu06} saying that no data structure for connectivity has both update- and query-time in $o(\log n)$. 
This is in stark contrast to the incremental situation, where only edge-insertions are allowed, in which the $\alpha$-time algorithm for union-find is tight~\cite{Tarjan:75,Fredman:89}. 
When restricted to deletions, however, even for general graphs, there are no known lower bounds beyond the trivial $O(|G|)$. 
The research in this paper is inspired by the fundamental open question of whether decremental (deletion-only) connectivity~\cite{Thorup97}, $2$-edge connectivity, biconnectivity, or even minimum cut for general graphs can be solved in amortized constant time per edge-deletion, or whether non-trivial lower bounds do exist. 

The following table (Figure~\ref{fig:table}) shows how we improve state-of-the-art for planar graphs and minor-free graphs. Here, we present maximum time \emph{per operation}, that is, we do not require $O(1)$ query time. When restricted to constant query time, the best biconnectivity algorithms for non-planar sparse graphs were fully dynamic and had an update time of $\tilde{O}(\sqrt n)$~\cite{Henzinger95}.

\begin{figure}[h]\label{fig:table}
	\begin{center}
\bgroup 
\def\arraystretch{1.5}%
\begin{tabular}{| l | c | c | c | c |}
\hline
& planar & bnd. genus & 
\multicolumn{2}{c|}{minor-free graphs} \\
& \emph{previous} & \emph{previous} & \multicolumn{1}{c}{\emph{previous}} & \textbf{now} \\
\hline
connectivity & $O(1)$~\cite{Lacki15} 
& $O(\log n)$~\cite{Eppstein:2003} &  $\tilde{O}(\log^2 n)$ \cite{Wulff-Nilsen16a}& $\mathbf{O(1)}$ \\
\hline
2-edge-connectivity & $O(1)$~\cite{Holm17} & \multicolumn{2}{r|}{ $\tilde{O}(\log ^2 n)$ \cite{Holm18a}} & $\mathbf{O(1)}$ \\
\hline
biconnectivity & $O(\log n)$~\cite{Holm17} &  \multicolumn{2}{r|}{ $\tilde{O}(\log^3 n )$ \cite{Holm18a}} & $\mathbf{O(1)}$  \\
\hline
\end{tabular}
\egroup 
	\end{center}
	\caption{Our improvements (\textbf{now}) in relation to previous results (\emph{previous}). The table shows amortized time per operation. 
	The table compares with state-of-the-art amortized deterministic algorithms. Allowing randomization, the best decremental connectivity algorithm runs in time $\tilde{O}(\log n)$~\cite{Thorup:2000}.}
\end{figure}

\paragraph{Dynamic graph connectivity} has been studied for decades. Most general is fully dynamic connectivity for general graphs~\cite{Frederickson85, HeTh97, Henzinger:1999, Holm:2001, Thorup:2000, Kapron:2013, HuangHKP17, Wulff-Nilsen16a, kejlbergrasmussen_et_al:LIPIcs:2016:6395, NSW17}, where edges are allowed to be both inserted and deleted. Similarly, fully dynamic two-edge connectivity and biconnectivity have been studied~\cite{Frederickson97,Henzinger95,Eppstein97,Henzinger97,Holm:2001,Thorup:2000,Holm18a} and have algorithms with polylogarithmic update- and query time. 

For special graph classes, such as planar graphs, graphs of bounded genus, and minor-free graphs, there has been a bulk of work on connectivity and higher connectivity, e.g.~\cite{EppItaTam92,Hershberger94,Giammarresi:96,Gustedt98,Eppstein99,Lacki2011,Lacki15,Holm17,Holm18b}. Our paper can be seen as a generalization of and improvement upon \cite{Lacki15}, who showed optimal amortized decremental connectivity for planar graphs, that is, amortized constant update time, and worst-case constant query time.

\paragraph{An $r$-division is,} intuitively, a family of $O(n/r)$ subgraphs  called the \emph{regions}, with $O(r)$ vertices each, such that the regions partition the edges, and each region shares %
$O(\sqrt{r})$ \emph{boundary vertices} with the rest of the graph.
The concept of $r$-divisions was introduced in~\cite{Frederickson:87}
as a tool for finding shortest paths in planar graphs. It naturally
generalizes the notion of a separator: a small set of vertices that
cause the graph to fall apart into \emph{two} regions, each containing
a constant fraction of the original graph~\cite{LiptonTarjan}.

Later, Henzinger et al.~\cite{HENZINGER19973} generalized this to the
concept of a \emph{strict $(r,s)$-division}, which is a family of
$O(n/r)$ subgraphs called the regions, with at most $r$ vertices
each, that partition the edges, and where each region has at most $s$
boundary vertices.  An $r$-division is thus a strict $(r,s)$-division with
$s=O(\sqrt{r})$.

\begin{sloppypar}
There are linear time algorithms for finding $r$-divisions for planar
graphs~\cite{Goodrich,Klein:2013} and
strict $(r,O(r^{1-\varepsilon}))$-divisions (for any sufficiently small
$\varepsilon>0$) for minor-free graphs using techniques
from~\cite{reed:hal-01184376,DBLP:journals/dam/TazariM09,Wulff-Nilsen11}\footnote{This result was first claimed by Henzinger et al.~\cite{HENZINGER19973}, but their solution only worked for planar or bounded degree $h$-minor-free graphs.  For general $h$-minor-free graphs the running time in~\cite{DBLP:journals/dam/TazariM09} is proportional to $2^{O(h^2)}$. This can be improved to $2^{O(h)}$ using techniques from~\cite{Wulff-Nilsen11}.}. See Appendix~\ref{sec:r-div}.
\end{sloppypar}

For the rest of our paper, we will often ignore these distinctions and just use the term $r$-division to mean any strict $(r,O(r^{1-\epsilon}))$-division for some suitable $r,\varepsilon$.

\paragraph{Our results.} We give a data structure for maintaining biconnectivity for a large class of graphs. In order to state our theorem in its fullest generality, we need to define what it means for a pair of $r$-divisions to be \emph{a suitable pair}.

Given a graph $G$ with $n$ vertices, we call a pair $(\mathcal{A},\mathcal{R})$ where $\mathcal{A}$ is a
strict $(r_1,s_1)$-division and $\mathcal{R}$ is a strict
$(r_2,s_2)$-division \emph{a suitable pair of $r$-divisions} if:
\begin{itemize}
\item there exists an algorithm for fully dynamic biconnectivity in general graphs with amortized time $t(n)$ per operation\footnote{e.g. $t(n)=O(\log^5n)$ using~\cite{Holm:2001}, and $t(n)=O(\log^3n\cdot\log^2\log n)$ using~\cite{Holm18a}}, such that:
\item each boundary vertex of $\mathcal{A}$ is also a boundary vertex of $\mathcal{R}$ ($\boundary\mathcal{A}\subseteq\boundary\mathcal{R}$); and
\item for each region $A\in\mathcal{A}$, $\mathcal{R}$ contains a
  partition of $A$ into $O(\frac{r_1}{r_2})$ regions of size at most
  $r_2$, each having at most $s_2$ boundary vertices\footnote{This is
    slightly weaker than requiring $\mathcal{R}$ to contain a strict
    $(r_2,s_2)$-division of $A$}; and
\item $r_1,s_1\in O(\poly(\log n))$ and $\frac{r_1}{s_1}\in \Omega(t(n)\log n)$; and
\item $r_2,s_2\in O(\poly(\log\log n))$ and $\frac{r_2}{s_2}\in \Omega(t(r_1)\log r_1)$.
\end{itemize}

Our data structure answers queries to biconnectivity, i.e, a pair of vertices are biconnected if they are connected and not separated by any bridge or cutvertex. If the vertices $u$ and $v$ are connected but not biconnected, we can output a cutvertex separating them, in fact, we can output that of the possibly many cutvertices that is \emph{nearest} to $u$ -- we call this the \emph{nearest cutvertex} -- or detect the special case where $uv$ is a bridge. 

\begin{theorem}\label{thm:main}
	There exists a data structure that given a graph $G$ with $n$ vertices and $m$ edges, and given a suitable pair of $r$-divisions,  preprocesses $G$ in $O(m+n)$ time and handles any series of edge-deletions in $O(m)$ total time while answering queries to pairwise biconnectivity and queries to nearest cutvertex in $O(1)$ time.
\end{theorem}

This can immediately be combined with any algorithm for finding suitable $r$-divisions in linear time, to obtain optimal decremental biconnectivity data structures for graphs that are e.g. planar, bounded genus, or minor free.

The data structure is easily extended to maintain information about connectivity, so as to answer queries to pairwise connectivity in $O(1)$ time, and our techniques can easily be used to obtain a decremental data structure for $2$-edge connectivity with the same update- and query times.

\subsection{Techniques}
Since the property of being an $r$-division is not violated as edges are deleted, it is natural to want to use $r$-divisions to get better decremental data structures for graphs. The idea is to have a top-level graph with size only proportional to the number of boundary vertices, and to handle the regions efficiently simply because they are smaller.

With biconnectivity, the first challenge is to design the top-level graph: a vertex may be not biconnected to any boundary vertex in its region, but yet be biconnected with some other vertex in an other region via two separate boundary vertices (see Figure~\ref{fig:sausages}). Even vertices from the same region may be biconnected in $G$ although they are not biconnected, or even connected, within the region.

\begin{figure}[h]
\begin{tikzpicture}[x=0.5cm,y=0.6cm,scale=0.25]
  \begin{scope}[
      every path/.style={
      },
      every node/.style={
        font=\tiny,
        text=white,
        inner sep=1pt,
      },
      every label/.style={
        label distance=2mm,
      },
      vertex set/.style={
        dashed,
      },
      vertex/.style={
        draw,
        circle,
        fill=white,
        minimum size=1.1mm,
        inner sep=0pt,
        outer sep=0pt,
      },
      boundary vertex/.style={vertex,fill=mygreen,draw=black!50!mygreen,minimum size=1.1mm,ultra thin},
      edge/.style={blue,thick},
      undirected edge/.style={edge},
      directed edge/.style={edge,->,>=stealth'},
      smalldot/.style={
      	draw,
      	circle,
      	fill=black,
      	text=black,
	    scale=1,
      	inner sep=0pt,
      	outer sep=0pt,
      },      
    ]

    \node[boundary vertex,label={below:$a$}] (a) at (-.5,-.5) {};
    \node[boundary vertex,label={below:$a'$}] (a') at (1.5,.5) {};
    \node[boundary vertex,label={below:$b$}] (b) at (-.2,-10.2) {};
    \node[boundary vertex,label={below:$b$}] (b') at (.2,-11.8) {};
    \node[vertex,label={below:$c$}] (c) at (4,-2) {};
    \node[vertex,label={below:$c'$}] (c') at (4,-10) {};
    \node[vertex,label={above:$d$}] (d) at (8,-3) {};
    \node[vertex,label={below:$d'$}] (d') at (8,-9) {};
    \node[vertex,label={above right:$e$}] (e) at (12,-4) {};
    \node[vertex,label={below:$e'$}] (e') at (12,-8) {};
    \node[vertex,label={below:$f$}] (f) at (16,-5) {};
    \node[vertex,label={below:$f'$}] (f') at (16,-7) {};
    \node[vertex,label={below:$g$}] (g) at (11,-9.5) {};
    \node[vertex,label={below:$h$}] (h) at (19,-6) {};
    \node[vertex,label={above:$j$}] (j) at (23,-6) {};
    \node[vertex,label={above:$k$}] (k) at (27,-6) {};
    \node[vertex,label={above:$l$}] (l) at (31,-6) {};
    \node[boundary vertex,label={above:$n$}] (n) at (34,-3) {};
    \node[vertex,label={above:$n'$}] (n') at (34,-9) {};
    \node[vertex,label={above:$o'$}] (o') at (38,-10) {};
    \node[vertex,label={above:$p'$}] (p') at (42,-11) {};
    \node[boundary vertex,label={above:$p$}] (p) at (45,-11.5) {};
    \node[vertex,label={above:$q$}] (q) at (13,-12) {};
    \node[vertex,label={above:$r$}] (r) at (16,-13.5) {};
    \node[boundary vertex,label={above:$s$}] (s) at (20,-23.5) {};
    \node[boundary vertex,label={above:$t$}] (t) at (21,-22) {};
    \node[vertex,label={above:$u$}] (u) at (23,-20) {};
    \node[vertex,label={above:$v$}] (v) at (26,-19) {};
    \node[vertex,label={above:$w$}] (w) at (30,-19) {};
    \node[vertex,label={above:$x$}] (x) at (33,-20) {};
    \node[boundary vertex,label={above:$y$}] (y) at (35.1,-21.3) {};
    \node[boundary vertex,label={above:$y$}] (y) at (34,-21.5) {};
	\node[smalldot] (X) at (14.6,-12.5) {};
	\node[text=black] (X') at (16.2,-11.6) {\normalsize $x$};
	\node[smalldot] (Y) at (25.4,-6.2) {};
	\node[text=black] (Y') at (26.5,-8) {\normalsize $y$};
	\node[smalldot] (Z) at (28,-18.7) {};
	\node[text=black] (Z') at (28.2,-17.1) {\normalsize $z$};

    \begin{scope}
      \path[use as bounding box] (-4,1.5) rectangle (30.3,-15.1);
      \draw (c) .. controls (-10,-4) and (4,7) .. (c);
      \draw (c) to[bend left=50] (d);
      \draw (d) to[bend left=50] (c);
      \draw (c') to[bend left=50] (d');
      \draw (d') to[bend left=50] (c');
      \draw (d) to[bend left=50] (e);
      \draw (e) to[bend left=50] (d);
      \draw (d') to[bend left=50] (e');
      \draw (e') to[bend left] (g);
      \draw (g) to[bend left] (d');
      \draw (e) to[bend left=50] (f);
	  \draw (f) to[bend left=50] (e);
      \draw (e') to[bend left=50] (f');
      \draw (f') to[bend left=50] (e');
      \draw (f) to[bend left] (h);
      \draw (h) to[bend left] (f');
      \draw (f') to[bend left] (f);
      \draw (h) to[bend left=50] (j);
      \draw (j) to[bend left=50] (h);
      \draw (k) to[bend left=50] (j);
      \draw (j) to[bend left=50] (k);
      \draw (k) to[bend left=50] (l);
      \draw (l) to[bend left=50] (k);
      \draw (l) to[bend left=50] (n);
      \draw (n) to[bend left=50] (l);
      \draw (l) to[bend left=50] (n');
      \draw (n') to[bend left=50] (l);
      \draw (n) .. controls (40,-4) and (38,3) .. (n);
      \draw (o') to[bend left=50] (n');
      \draw (n') to[bend left=50] (o');
      \draw (o') to[bend left=50] (p');
      \draw (p') to[bend left=50] (o');
      \draw (p') .. controls (47,-9) and (48,-15) .. (p');
      \draw (g) to[bend left=60] (q);
      \draw (q) to[bend left=60] (g);
      \draw (r) to[bend left=50] (q);
      \draw (q) to[bend left=50] (r);
      \draw (q) .. controls (9,-14) and (13,-17) .. (q);
      \draw (r) .. controls (20,-14) and (17,-19) .. (r);
      \draw (t) .. controls (21,-27) and (16,-22) .. (t);
      \draw (t) to[bend left=60] (u);
      \draw (u) to[bend left=60] (t);
      \draw (u) to[bend left=50] (v);
      \draw (v) to[bend left=50] (u);
      \draw (w) to[bend left=45] (v);
      \draw (v) to[bend left=45] (w);
      \draw (w) to[bend left=45] (x);
      \draw (x) to[bend left=45] (w);
      \draw (x) .. controls (40,-21) and (33,-25) .. (x);
      \draw (c') .. controls (-5,-6) and (-2,-18) .. (c');
    \end{scope}
  \end{scope}
\end{tikzpicture}%
        \hfill%
\begin{tikzpicture}[x=0.5cm,y=0.6cm,scale=0.25]
  \begin{scope}[
      every path/.style={
      },
      every node/.style={
        font=\tiny,
        text=white,
        inner sep=1pt,
      },
      every label/.style={
        label distance=2mm,
      },
      vertex set/.style={
        dashed,
      },
      vertex/.style={
        draw,
        circle,
        fill=white,
        minimum size=1.1mm,
        inner sep=0pt,
        outer sep=0pt,
      },
      boundary vertex/.style={vertex,fill=mygreen,draw=black!50!mygreen,minimum size=1.1mm,ultra thin},
      edge/.style={blue,thick},
      undirected edge/.style={edge},
      directed edge/.style={edge,->,>=stealth'},
      smalldot/.style={
      	draw,
      	circle,
      	fill=black,
      	text=black,
	    scale=1,
      	inner sep=0pt,
      	outer sep=0pt,
      },      
    ]

    \node[boundary vertex,label={below:$a$}] (a) at (-.5,-.5) {};
    \node[boundary vertex,label={below:$a'$}] (a') at (1.5,.5) {};
    \node[boundary vertex,label={below:$b$}] (b) at (-.2,-10.2) {};
    \node[boundary vertex,label={below:$b$}] (b') at (.2,-11.8) {};
    \node[vertex,label={below:$c$}] (c) at (4,-2) {};
    \node[vertex,label={below:$c'$}] (c') at (4,-10) {};
    \node[vertex,label={below:$f$}] (f) at (16,-5) {};
    \node[vertex,label={below:$f'$}] (f') at (16,-7) {};
    \node[vertex,label={below:$h$}] (h) at (19,-6) {};
    \node[vertex,label={above:$l$}] (l) at (31,-6) {};
    \node[boundary vertex,label={above:$n$}] (n) at (34,-3) {};
    \node[boundary vertex,label={above:$p$}] (p) at (45,-11.5) {};
    \node[boundary vertex,label={above:$s$}] (s) at (20,-23.5) {};
    \node[boundary vertex,label={above:$t$}] (t) at (21,-22) {};
    \node[vertex,label={above:$x$}] (x) at (33,-20) {};
    \node[boundary vertex,label={above:$y$}] (y) at (35.1,-21.3) {};
    \node[boundary vertex,label={above:$y$}] (y) at (34,-21.5) {};
	\node[smalldot] (X) at (14.6,-12.5) {};
	\node[text=black] (X') at (16.2,-11.6) {\normalsize $x$};
	\node[smalldot] (Y) at (25.4,-6.2) {};
	\node[text=black] (Y') at (26.5,-8) {\normalsize $y$};
	\node[smalldot] (Z) at (28,-18.7) {};
	\node[text=black] (Z') at (28.2,-17.1) {\normalsize $z$};

    \begin{scope}
      \path[use as bounding box] (-4,1.5) rectangle (30.3,-15.1);
      \draw (c) .. controls (-10,-4) and (4,7) .. (c);
      \draw[dashed] (c) to[bend left=20] (f);
      \draw[dashed] (f) to[bend left=20] (c);
      \draw[dashed] (c') to[bend left=20] (f');
      \draw[dashed] (f') to[bend left=20] (c');
      \draw (f) to[bend left] (h);
      \draw (h) to[bend left] (f');
      \draw (f') to[bend left] (f);
      \draw[dashed] (h) to[bend left=20] (l);
      \draw[dashed] (l) to[bend left=20] (h);
      \draw (l) to[bend left=50] (n);
      \draw (n) to[bend left=20] (l);
      \draw[dashed] (l) to[bend right=10] (41.6,-12.5) .. controls (48,-15) and (47,-9.5)  .. (44.5,-9.5) to[bend right=10] (l);
      \draw (t) .. controls (21,-27) and (16,-22) .. (t);
      \draw[dashed] (t) to[bend left=60] (x);
      \draw[dashed] (x) to[bend right=10] (t);
      \draw (x) .. controls (40,-21) and (33,-25) .. (x);
      \draw (c') .. controls (-5,-6) and (-2,-18) .. (c');
    \end{scope}
  \end{scope}
\end{tikzpicture}%
	\caption{{\small Left: A region $R$ with $10$ boundary vertices (green). There is a vertex separating $x$ from the boundary, so $x$ is never biconnected with anything in $G\setminus R$. Vertices $y$ and $z$ however, are not even connected in $R$, but may be biconnected in $G$. Right: The structure may be compressed in the sense depicted: $x$ is not represented at all, while $y$ and $z$ are represented in pseudo-blocks (dashed).} \label{fig:sausages}}
\end{figure}
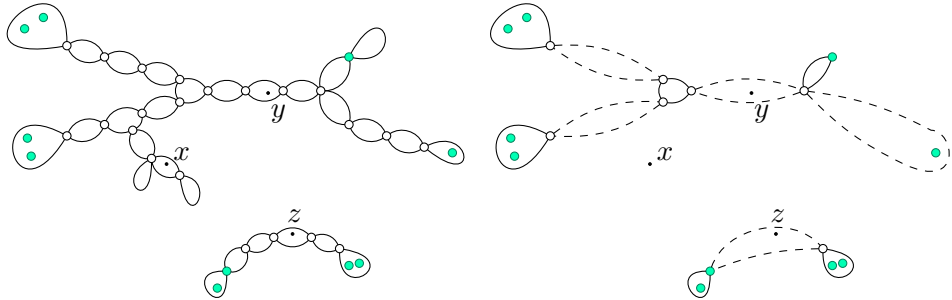

We thus need to store an efficient representation of the biconnectivity of the region as seen from the perspective of the boundary vertices. We call this efficient representation the \emph{compressed BC-forest} (see Section~\ref{sec:patchwork}). It is obtained from the forest of BC-trees (also known as the block-cutpoint trees, see Section~\ref{sec:prelims}) by first marking certain blocks and cutvertices as \emph{critical}, and then, basically, contracting the paths that connect them. The critical blocks and cutvertices are spartanly chosen, such that the total size of all the compressed BC-forests is only proportional to the boundary itself. We stitch the compressed BC-forests together by the boundary vertices they share, and obtain the \emph{patchwork graph} (see Figure~\ref{fig:patchwork}), in which all vertices that are biconnected to anything outside their region are \emph{represented}, and we use the \emph{representatives} of vertices to reveal when they are biconnected by paths that go via boundary vertices.  A construction very similar to our compressed BC-forests appears in~\cite{Galil:1999}, where it is used in a separator tree for a planar graph, but the rules for what to contract are subtly different.

\begin{figure}[h]
\centering
\resizebox{0.8\textwidth}{!}{%
\begin{tikzpicture}[x=1cm,y=1cm,scale=0.5]
  \begin{scope}[
      every path/.style={
      },
      every node/.style={
        font=\tiny,
        text=white,
        inner sep=1pt,
      },
      every label/.style={
        label distance=2mm,
      },
      vertex set/.style={
        dashed,
      },
      vertex/.style={
        draw,
        circle,
        fill=white,
        minimum size=0.8mm,
        inner sep=0pt,
        outer sep=0pt,
      },
      boundary vertex/.style={vertex,fill=mygreen,draw=black!50!green,minimum size=1mm,ultra thin},
      edge/.style={blue,thick},
      undirected edge/.style={edge},
      directed edge/.style={edge,->,>=stealth'},
      smalldot/.style={
      	draw=black!50!mygreen,
      	circle,
      	fill=mygreen,
      	text=black,
      	line width=0.02mm,
	    scale=1,
      	inner sep=0pt,
      	outer sep=0pt,
      	minimum size=.5mm,
      },   
      smallnode/.style={
      	draw=black,
      	rectangle,
      	fill=white,
      	text=black,
      	line width=0.05mm,
      	scale=1.2,
      	inner sep=0pt,
      	outer sep=0pt,
      },   
      smallround/.style={
      	draw=black,
      	circle,
      	fill=white,
      	text=black,
      	line width=0.05mm,
      	scale=1,
      	inner sep=0pt,
      	outer sep=0pt,
      },   
      mydash/.style={
      	dash pattern=on 2pt off 1pt,
	  },   
    ]

    \begin{scope}[scale=0.25]
    \draw[black,fill=gray!15,line width=0.05mm,rounded corners,use as bounding box] (0,0) .. controls (-1,-4) and (.5,-8) .. (0,-11) 
    .. controls (3,-12) and (15, -10) .. (20,-11) .. controls (22,-9) and (19.5,-4) ..  (20,0) .. controls (17,1.5) and (14,3) .. (7,0) -- cycle;  
    \draw[mydash,line width=0.05mm] (-.38,-4) -- (20.03,-4);
    \draw[mydash,line width=0.05mm] (.1,-8) -- (21,-8);
    \draw[mydash,line width=0.05mm] (6.1,0) -- (19.55,0);
    \draw[mydash,line width=0.05mm] (5,0) -- (5,-11.4);
    \draw[mydash,line width=0.05mm] (10.3,1.25) -- (10.3,-10.9);
    \draw[mydash,line width=0.05mm] (15.6,1.75) -- (15.6,-10.65);

    \node[smalldot] (b1) at (5,-2) {};
    \node[smalldot] (b2) at (5,-4) {};
    \node[smalldot] (b3) at (2.8,-4) {};
    \node[smalldot] (b4) at (5,-8) {};
    \node[smalldot] (b5) at (7,-4) {};
    \node[smalldot] (b12) at (8,-4) {};
    \node[smalldot] (b6) at (8,-8) {};
    \node[smalldot] (b7) at (10.3,-8) {};
    \node[smalldot] (b13) at (10.3,-5) {};
    \node[smalldot] (b8) at (10.3,-10) {};
    \node[smalldot] (b9) at (5,-10.5) {};
    \node[smalldot] (b10) at (10.3,-2.5) {};
    \node[smalldot] (b11) at (10.3,-1) {};
    \node[smalldot] (b14) at (1.1,-4) {};
    \node[smalldot] (b15) at (5,-6) {};
    \node[smalldot] (b16) at (3.4,-8) {};
    \node[smalldot] (b17) at (1,-8) {};
    \node[smalldot] (b18) at (12,0) {};
    \node[smalldot] (b19) at (14,0) {};
    \node[smalldot] (b20) at (12.5,-4) {};
    \node[smalldot] (b21) at (15.6,0) {};
    \node[smalldot] (b22) at (15.6,1.5) {};
    \node[smalldot] (b23) at (5,-9.4) {};
    \node[smalldot] (b24) at (14.5,-4) {};
    \node[smalldot] (b25) at (15.6,-5) {};
    \node[smalldot] (b26) at (12.5,-8) {};
    \node[smalldot] (b27) at (14,-8) {};
    \node[smalldot] (b28) at (15.6,-9) {};
    \node[smalldot] (b29) at (15.6,-10) {};
    \node[smalldot] (b30) at (15.6,-2) {};
    \node[smalldot] (b31) at (18,-4) {};
    \node[smalldot] (b32) at (17.2,-8) {};
    \node[smalldot] (b33) at (19.2,-8) {};

    \node[smallnode] (i0) at (2.5,-1.5) {};
	\draw[line width=0.05mm] (b1) -- (i0) node[smallnode,near start]{} node[smallround,near end]{};
	\draw[line width=0.05mm] (b3) -- (i0) node[smallnode,near start]{} node[smallround,near end]{};
	\draw[line width=0.05mm] (b14) -- (i0);%

	\node[smallnode] (i6) at (2.8,-6) {};
	\draw[line width=0.05mm] (b14) -- (b17) node[smallnode,midway]{};
	\draw[line width=0.05mm] (b14) -- (i6) node[smallnode,near start]{} node[smallround,near end]{};
	\draw[line width=0.05mm] (b15) -- (i6);
	\draw[line width=0.05mm] (b16) -- (i6) node[smallnode,near start]{} node[smallround,near end]{};

	\node[smallnode] (i8) at (1.3,-9.5) {};
	\draw[line width=0.05mm] (b23) -- (i8) node[smallnode,near start]{} node[smallround,near end]{};
	\draw[line width=0.05mm] (b4) -- (i8) node[smallnode,near start]{} node[smallround,near end]{};
	\draw[line width=0.05mm] (b17) -- (i8);

    \node[smallnode] (i1) at (6.8,-1.5) {};
	\draw[line width=0.05mm] (b2) -- (i1) node[smallnode,midway]{} node[smallround,near end]{};
	\draw[line width=0.05mm] (b1) -- (i1);
	\draw[line width=0.05mm] (b5) -- (i1) node[smallnode,near start]{} node[smallround,near end]{};
    \node[smallnode] (i2) at (8.2,-3) {};
	\draw[line width=0.05mm] (i1) -- (i2) node[smallround,midway]{};
	\draw[line width=0.05mm] (b10) -- (i2) node[smallnode,near start]{} node[smallround,near end]{};
	\draw[line width=0.05mm] (b11) -- (i2) node[smallnode,near start]{} node[smallround,near end]{};

	\node[smallnode] (i3) at (7,-6.5) {};
	\draw[line width=0.05mm] (b2) -- (i3) node[smallnode,midway]{} node[smallround,near end]{};
	\draw[line width=0.05mm] (b4) -- (i3) node[smallnode,midway]{} node[smallround,near end]{};
	\draw[line width=0.05mm] (b12) -- (i3) node[smallnode,midway]{} node[smallround,near end]{};
	\node[smallnode] (i4) at (8.5,-6.5) {};
	\draw[line width=0.05mm] (b6) -- (i4);
	\draw[line width=0.05mm] (b7) -- (i4) node[smallnode,near start]{} node[smallround,near end]{};
	\draw[line width=0.05mm] (b13) -- (i4) node[smallnode,near start]{} node[smallround,near end]{};

	\node[smallnode] (i9) at (7,-10) {};
	\draw[line width=0.05mm] (b9) -- (i9) node[smallnode,near start]{} node[smallround,near end]{};
	\draw[line width=0.05mm] (b6) -- (i9) node[smallnode,near start]{} node[smallround,near end]{};
	\draw[line width=0.05mm] (b8) -- (i9) node[smallnode,midway]{} node[smallround,near end]{};

	\node[smallnode] (i7) at (13,1.2) {};
	\draw[line width=0.05mm] (b18) -- (i7);
	\draw[line width=0.05mm] (b19) -- (i7);
	\draw[line width=0.05mm] (b21) -- (i7) node[smallnode,midway]{} node[smallround,near end]{};
	\draw[line width=0.05mm] (b22) -- (i7) node[smallnode,near start]{} node[smallround,near end]{};

	\draw[line width=0.05mm] (b10) -- (b18) node[smallnode,midway]{};
	\draw[line width=0.05mm] (b20) -- (b19) node[smallnode,midway]{};

	\node[smallnode] (i10) at (12,-5) {};
	\node[smallnode] (i11) at (14.5,-5) {};
	\node[smallnode] (i12) at (13.2,-6) {};
	\node[smallnode] (i13) at (13.3,-7) {};
	\draw[line width=0.05mm] (i10) -- (i12) node[smallround,midway]{};
	\draw[line width=0.05mm] (i11) -- (i12) node[smallround,midway]{};
	\draw[line width=0.05mm] (i12) -- (i13) node[smallround,midway]{};
	\draw[line width=0.05mm] (b26) -- (i13);
	\draw[line width=0.05mm] (b27) -- (i13);
	\draw[line width=0.05mm] (b13) -- (i10)node[smallnode,near start]{} node[smallround,near end]{};
	\draw[line width=0.05mm] (b20) -- (i10);
	\draw[line width=0.05mm] (b24) -- (i11);
	\draw[line width=0.05mm] (b25) -- (i11);

	\draw[line width=0.05mm] (b8) -- (b26) node[smallnode,midway]{};
	\draw[line width=0.05mm] (b27) -- (b29) node[smallnode,midway]{};

	\node[smallnode] (i14) at (17.3,-1.5) {};
	\draw[line width=0.05mm] (b30) -- (i14) node[smallnode,near start]{} node[smallround,near end]{};
	\draw[line width=0.05mm] (b31) -- (i14) node[smallnode,midway]{} node[smallround,near end]{};
	\draw[line width=0.05mm] (b21) -- (i14) node[smallnode,near start]{} node[smallround,near end]{};

	\draw[line width=0.05mm] (b31) -- (b32) node[smallnode,midway]{};
	\draw[line width=0.05mm] (b32) -- (b25) node[smallnode,midway]{};

	\draw[line width=0.05mm] (b28) -- (b33) node[smallnode,midway]{};

    \end{scope}
  \end{scope}
\end{tikzpicture}%
}
\caption{\small An $r$-division and its corresponding patchwork graph. The graph is bipartite between, on one hand, round boundary vertices and cutvertices, and, on the other hand, square blocks and contracted (pseudo) blocks. \label{fig:patchwork}}
\end{figure}
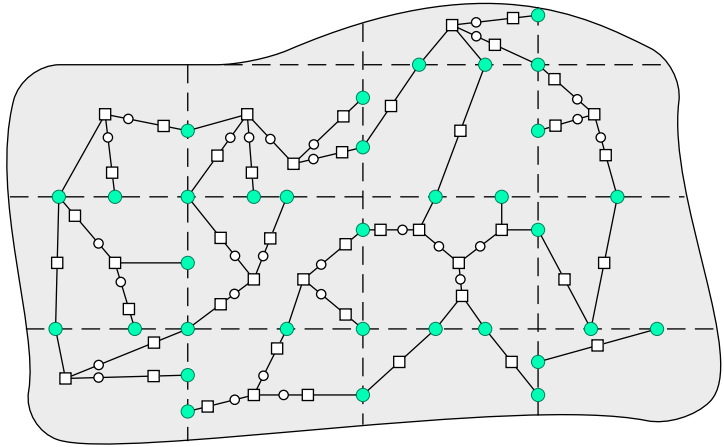

If decremental changes to a region only gave rise to decremental changes to its forest of BC-trees, we would be close to done.
However, and this is the second challenge, the deletion of an edge can cause a block to fall apart into a chain of blocks. Luckily, the damage to the compressed BC-forest is containable: 
only %
$O(n/\polylog n)$ 
vertices can be present in the compressed BC-forest, and the changes can be modeled by only three operations: edge- or path deletions, certain forms of vertex splits, and contractions of paths. These operations, we show, are of a form that can be handled in polylogarithmic time by one of the fully-dynamic biconnectivity data structures (see Section~\ref{sec:patchwork}). 

While using $r$-divisions once would obtain an improvement from polylog to polyloglog, which might, in practice, be useful already, it is tempting to form $r$-divisions of the regions themselves and use recursion in order to obtain an even faster speedup (see Figure~\ref{fig:layers}). This would mean that each region should again contain a patchwork made from the compressed BC-forests of its subregions (and, luckily, these patchwork operations compose beautifully). Thus, via recursion, one can obtain a purely combinatorial data structure with $O(\log^{\ast}n)$ update- and query time. But in fact, with standard RAM-tricks, if the subregions are of only polyloglog size, one can handle any operation in constant time -- simply by using a look-up table. Thus, in the practical RAM-model (i.e. the RAM-model with standard AC$^0$ operations such as addition, subtraction, bitwise and/or/xor), we can make do with only $3$ levels (top, middle and bottom), and obtain $O(1)$ update- and query-time.

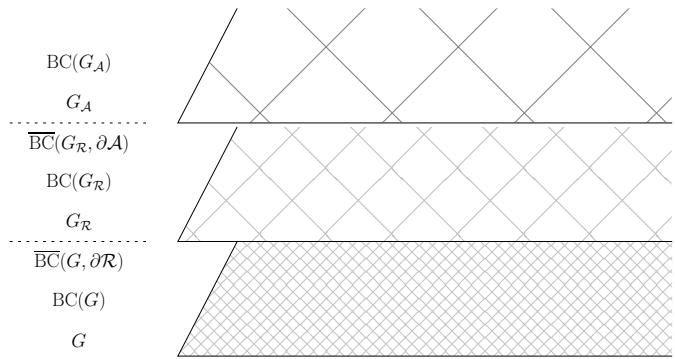
\begin{figure}[h]
	\centering
	\resizebox{0.7\textwidth}{!}{%
\begin{tikzpicture}[x=0.5cm,y=0.5cm,scale=1.5]
  \begin{scope}[
      every path/.style={
      },
      every node/.style={
        font=\huge,
        text=black,
        inner sep=1pt,
      },
      every label/.style={
        label distance=5mm,
      },
      vertex set/.style={
        dashed,
      },
      vertex/.style={
        draw,
        circle,
        fill=white,
        minimum size=2mm,
        inner sep=0pt,
        outer sep=0pt,
      },
      edge/.style={blue,thick},
      undirected edge/.style={edge},
      directed edge/.style={edge,->,>=stealth'},
    ]

	\node (bcG0) at (0,14) {$\BC(G_{\mathcal{A}})$};
    \node (G0) at (0,12) {$G_{\mathcal{A}}$};
    \node (aG0) at (0,13) {};
	\draw[loosely dashed] (-3.5,11) to (3.5,11);
    \node (cbcG1) at (0,10) {$\CBC(G_{\mathcal{R}},\partial\mathcal{A})$};
    \node (bcG1) at (0,8) {$\BC(G_{\mathcal{R}})$};
    \node (G1) at (0,6) {$G_{\mathcal{R}}$};
    \node (aG1) at (0.1,7) {};
	\draw[loosely dashed] (-3.5,5) to (3.5,5);
    \node (cbcG) at (0,4) {$\CBC(G,\partial \mathcal{R})$};
    \node (bcG) at (0,2) {$\BC(G)$};
    \node (G) at (0,0) {$G$};
    \node (aG) at (0.3,1.2) {};

    \draw[draw=white,pattern=custom north west lines,hatchspread=50pt,hatchcolor=gray] let \p1=(aG0) in (8,\y1+1.9cm) -- (5,\y1-1cm) -- (30,\y1-1cm) -- (30,\y1+1.9cm);
    \draw[draw=white,pattern=custom north east lines,hatchspread=50pt,hatchcolor=gray] let \p1=(aG0) in (8,\y1+1.9cm) -- (5,\y1-1cm) -- (30,\y1-1cm) -- (30,\y1+1.9cm);
    
	\draw[draw=white,pattern=custom north west lines,hatchspread=25pt,hatchcolor=gray!50,hatchshift=4pt] let \p1=(aG1) in (8,\y1+1.9cm) -- (5,\y1-1cm) -- (30,\y1-1cm) -- (30,\y1+1.9cm);
    \draw[draw=white,pattern=custom north east lines,hatchspread=25pt,hatchcolor=gray!50,hatchshift=20pt] let \p1=(aG1) in (8,\y1+1.9cm) -- (5,\y1-1cm) -- (30,\y1-1cm) -- (30,\y1+1.9cm);

	\draw[draw=white,pattern=custom north west lines,hatchspread=5pt,hatchcolor=gray!45,dotted] let \p1=(aG) in (8,\y1+1.9cm) -- (5,\y1-1cm) -- (30,\y1-1cm) -- (30,\y1+1.9cm);
	\draw[draw=white,pattern=custom north east lines,hatchspread=5pt,hatchcolor=gray!45,dotted] let \p1=(aG) in (8,\y1+1.9cm) -- (5,\y1-1cm) -- (30,\y1-1cm) -- (30,\y1+1.9cm);
        
    \foreach \i/\g in {%
      0/aG0,1/aG1,2/aG%
    }{%
      \draw let \p1=(\g) in (8,\y1+1.9cm) -- (5,\y1-1cm) -- (30,\y1-1cm);
    }

  \end{scope}
\end{tikzpicture}
}
	\caption{\small We use nested $r$-divisions and obtain a levelled structure. Each level maintains a graph, its BC-tree, and, for the non-top levels, the compressed BC-tree with relation to the boundary.\label{fig:layers}}
\end{figure}

Here, as our third challenge, we face that one does not simply recurse into optimality -- we need to assure ourselves that when a deletion of an edge causes changes in the compressed BC-trees of the subregion, the changes to the patchwork graph on the level above are manageable. Here, we show that our carefully chosen forms of vertex splits and path contractions do indeed only give rise to the same variant of splits and contractions on the parent level. 

Finally, when a pair of vertices $u,v$ are connected but not biconnected, we can in constant time find the nearest cutvertex on any path from $u$ to $v$ -- this is called the \emph{nearest cutvertex problem} (see Figure~\ref{fig:nearest}). We show that the nearest cutvertex can be determined by at most one nearest cutvertex and one biconnected query in the patchwork graph, and at most one nearest cutvertex and one biconnected query in the region. We also show how to augment an explicit representation of the BC-tree subject to certain splits,  contractions, and deletions such that we can still access the nearest cutvertex - a problem that reduces to first-on-path on a dynamic tree subject to certain vertex splits, and certain edge contractions and deletions. We solve this by solving a seemingly harder problem on such trees, namely that of answering an extended form of the nearest common ancestor query, known as the \emph{characteristic ancestor} query.  This solution may be of independent interest.
\begin{figure}[h]
\centering
\begin{tikzpicture}[x=0.5cm,y=0.5cm,scale=0.5]
  \begin{scope}[
      every path/.style={
      },
      every node/.style={
        font=\tiny,
        text=white,
        inner sep=1pt,
      },
      every label/.style={
        label distance=2mm,
      },
      vertex set/.style={
        dashed,
      },
      vertex/.style={
        draw,
        circle,
        fill=white,
        minimum size=0.8mm,
        inner sep=0pt,
        outer sep=0pt,
      },
      boundary vertex/.style={vertex,fill=mygreen,draw=black,minimum size=0.9mm,ultra thin},
      edge/.style={blue,thick},
      undirected edge/.style={edge},
      directed edge/.style={edge,->,>=stealth'},
      smalldot/.style={
      	draw,
      	circle,
      	fill=black,
      	text=black,
	    scale=1,
      	inner sep=0pt,
      	outer sep=0pt,
      },      
    ]

    \node[vertex,label={below:$a$}] (a) at (0,0) {};
    \node[vertex,label={below:$b$}] (b) at (2.5,2.5) {};
    \node[vertex,label={below:$c$}] (c) at (6,6) {};
    \node[vertex,draw=red,minimum size=0.4mm] (x) at (3.5,3.1) {};
    \node[vertex,draw=red,minimum size=0.4mm] (y) at (5.2,4.2) {};
    \begin{scope}
      \path[use as bounding box] (-3.5,-2.5) rectangle (15,5.1);

      \draw (a) to[bend left=40] (b);
      \draw (b) to[bend left=40] (a);
      \draw (b) to[bend left=60] (c);
      \draw (c) to[bend left=60] (b);
      \draw (a) .. controls (0,-5) and (-5,0) .. (a);
      \draw (c) to[bend left] (7,8);
      \draw (c) to[bend right] (8,7);
      \draw[red] (x) to[bend right] (y);

	\node[smalldot] (X) at (-1.5,-1.5) {};
	\node[text=black] (X') at (-1,-1) {\normalsize $x$};
	\node[text=black] (r) at (7.5,7.5) {\rotatebox{45}{\normalsize $\ldots$ \rotatebox{-45}{$r$}}};
	\node[smalldot] (Y) at (4.2,4.2) {};
	\node[text=black] (Y') at (4.9,4.8) {\normalsize $y$};

    \end{scope}
  \end{scope}
\end{tikzpicture}%
\hfil %
\begin{tikzpicture}[x=0.5cm,y=0.5cm,scale=0.5]
  \begin{scope}[
      every path/.style={
      },
      every node/.style={
        font=\tiny,
        text=white,
        inner sep=1pt,
      },
      every label/.style={
        label distance=2mm,
      },
      vertex set/.style={
        dashed,
      },
      vertex/.style={
        draw,
        circle,
        fill=white,
        minimum size=0.8mm,
        inner sep=0pt,
        outer sep=0pt,
      },
      boundary vertex/.style={vertex,fill=mygreen,draw=black,minimum size=0.9mm,ultra thin},
      edge/.style={blue,thick},
      undirected edge/.style={edge},
      directed edge/.style={edge,->,>=stealth'},
      smalldot/.style={
      	draw,
      	circle,
      	fill=black,
      	text=black,
	    scale=1,
      	inner sep=0pt,
      	outer sep=0pt,
      },      
    ]

    \node[vertex,label={below:$a$}] (a) at (0,0) {};
    \node[vertex,label={below:$b$}] (b) at (2.5,2.5) {};
    \node[vertex,label={below:$b''$}] (b2) at (4.5,2) {};
    \node[vertex,label={below:$d$}] (d) at (7,-.5) {};
    \node[vertex,label={below:$b'$}] (b1) at (4,4) {};
    \node[vertex,label={below:$c$}] (c) at (6,6) {};
    \node[vertex,label={below:$c'$}] (c1) at (8,8) {};
    \node[vertex,draw=red,minimum size=0.4mm] (x) at (7.3,7.1) {};
    \node[vertex,draw=red,minimum size=0.4mm] (y) at (8.2,-1.1) {};

    \begin{scope}
      \path[use as bounding box] (-3.5,-1.5) rectangle (15,7.1);
      \draw (a) to[bend left=40] (b);
      \draw (b) to[bend left=40] (a);
      \draw (b) to[bend left=40] (b1);
      \draw (b2) to[bend left=40] (b);
      \draw (b1) to[bend left=40] (b2);
      \draw (b1) to[bend left=50] (c);
      \draw (c) to[bend left=50] (b1);
      \draw (c) to[bend left=50] (c1);
      \draw (c1) to[bend left=50] (c);
      \draw (d) to[bend left=40] (b2);
      \draw (b2) to[bend left=40] (d);
      \draw (a) .. controls (0,-5) and (-5,0) .. (a);
      \draw (d) .. controls (7,-5) and (12,-.5) .. (d);
      \draw (c1) to[bend left] (8.5,9.5);
      \draw (c1) to[bend right] (9.5,8.5);
      \draw[red,dashed] (x) to[bend left] (y);

	\node[smalldot] (X) at (-1.5,-1.5) {};
	\node[text=black] (X') at (-1,-1) {\normalsize $x$};
	\node[text=black] (r) at (9.8,9.8) {\rotatebox{45}{\normalsize $\ldots$ \rotatebox{-45}{$r$}}};
	\node[smalldot] (Y) at (4.7,4.7) {};
	\node[text=black] (Y') at (5.4,5.3) {\normalsize $y$};

    \end{scope}
  \end{scope}
\end{tikzpicture}%
 \caption{An edge-deletion (red) in the graph can lead to a split of a block which changes the nearest cutvertex from $y$ towards $x$.\label{fig:nearest}}
\end{figure}
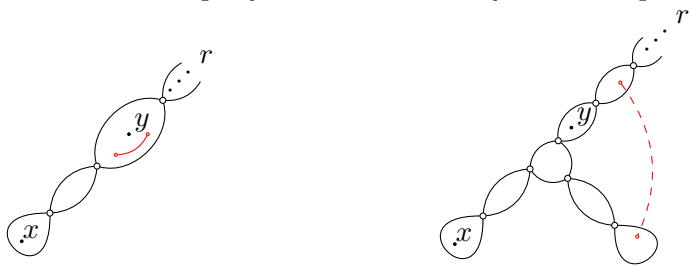

\paragraph{Related techniques} The idea of using recursive separators stems from Eppstein's sparsification technique~\cite{Eppstein97,Eppstein99}, where it secured $O(\sqrt{n})$ update algorithms for a series of problems, and the idea of using two levels of regions of size $O(\polylog n)$ and $O(\poly\log\log n)$, respectively, was introduced in~\cite{Lacki15} where the idea, together with a union-find structure in the dual graphs, was used to obtain amortized $O(1)$ decremental connectivity for planar graphs.

\paragraph{Paper outline}
Section~\ref{sec:prelims} is dedicated to preliminaries and terminology. Then, in Section~\ref{sec:capacitated}, we introduce the notion of \emph{capacitated biconnectivity}, which is a tool for overcoming the third challenge of making the recursion work. Section~\ref{sec:patchwork} is dedicated to an understanding of the patchwork graph in a static setting: how it is defined, how it reflects biconnectivity, and how it behaves when there is not one but two or more nested $r$-divisions of the same graph. In Section~\ref{sec:decremental}, we show how to maintain the patchwork graph decrementally, thus enabling us solve decremental biconnectivity. Finally, in Section~\ref{sec:nearest-cutvertex}, we show how to handle nearest cutvertex queries using our characteristic ancestors structure, which is described in Appendix~\ref{sec:first-on-path}.

\section{Preliminaries}\label{sec:prelims}
Given a graph with vertices $u$ and $v$, we say they are \emph{connected} if there is a path connecting them. A pair of connected vertices are \emph{$2$-edge connected} unless there is an edge whose removal would disconnect them. Such an edge is called a \emph{bridge}. 
A pair of $2$-edge connected vertices $u$ and $v$ are \emph{(locally) biconnected} unless there exists a vertex (other than $u$ and $v$) whose removal would disconnect them. Such a vertex is called a \emph{cutvertex}. 
For an ordered pair $(u,v)$ of connected but not biconnected vertices, 
the \emph{nearest cutvertex} separating them is uniquely defined as the first cutvertex on a path -- any path -- from $u$ to $v$. In the special case where $u$ and $v$ are separated by the bridge $uv$, we say that the nearest cutvertex is $\nil$. 

The \emph{blocks} of a graph are the maximal biconnected subgraphs. Each block is either a bridge or a maximal set of biconnected vertices.
For each connected component of a graph, the \emph{block-cutpoint tree}~\cite[p. 36]{Harary69}, or \emph{BC-tree} for short, reflects the biconnectivity among the vertices.
This tree has all the vertices of the graph and, furthermore, a vertex for each block. Its edges are those that connect each vertex to the block or blocks it belongs to.
If the graph $G$ is not necessarily connected, its \emph{forest of BC-trees} $\BC(G)$ has a BC-tree for each connected component of the graph.
The forest of BC-trees of a graph can be found in linear time~\cite{Tarjan72}.

If each BC-tree in the forest of BC-trees is rooted at an arbitrary block, each non-root block has one unique cutvertex separating it from its parent. Then, a pair of vertices are biconnected if and only if they either have the same non-bridge block as parent, or one is the parent of the non-bridge block that is parent of
the other. %

A dynamic data structure for biconnectivity in general graphs is developed in~\cite{Holm:2001,Thorup:2000,Holm18a}; it maintains an $n$-vertex graph and handles deletions and insertions of edges in $t(n)=O(\log^3 n\cdot \log^2\log n)$ amortized time, and answers queries in $O(\log^2 n\cdot \log^2\log n)$ worst-case time. The data structure is easily modified to give the first cutvertex separating a pair of vertices in $O(\log^2 n\cdot\log^2\log n)$ time, but even without this modification, one can find the first cutvertex via a binary search along a spanning tree in $O(\log n)$ queries in $O(\log^3n\cdot\log^2\log n)$ worst case time. 
Note however that for our purposes, the original~\cite{Holm:2001} data structure with $O(\log^5 n)$ amortized update- and query time is sufficient.
For the rest of this paper, we will just use $t(n)$ to denote the amortized time per operation (queries included) of a fully dynamic biconnectivity structure for general graphs.

For (not necessarily distinct) vertices $v,u,w$ in a tree, we use $v \longleftrightarrow u$ to denote the tree-path connecting $v$ and $u$, and we use $\meet(u,v,w)$ to denote the unique common vertex of all three tree-paths connecting them.

A \emph{strict $(r,s)$-division} is a set of $O(n/r)$ subgraphs $\mathcal{R}=\{R_1,R_2,\ldots \}$ called \emph{regions}, that partition the edges. Each region $R\in \mathcal{R}$ has at most $r$ vertices, and a set $\partial R$ of at most $s$ \emph{boundary vertices}, such that only boundary vertices appear in more than one region. We denote by $\partial \mathcal{R}$ the set of all boundary vertices $\bigcup_{R\in \mathcal{R}} \partial R$. 
Note that with these definitions, $\sum_{R\in \mathcal{R}} \abs{\partial R} \leq O(n/r)\cdot O(s) = O(n\cdot\frac{s}{r})$.

An $r$-division usually means a strict $(r,s)$-division with $s=O(\sqrt{r})$, but we will be using it more broadly to include any strict $(r,s)$-division, where $s=O(r^{1-\varepsilon})$ for some $\varepsilon>0$.

We say that a pair $(\mathcal{A},\mathcal{R})$ consisting of an $r_1$-division and an $r_2$-division are \emph{nested}, if $\partial \mathcal{A} \subseteq \partial \mathcal{R}$, and $\mathcal{R}$ contains an $r_2$-division of each region of $\mathcal{A}$.  With a slight abuse of notation, for any $A\in \mathcal{A}$ we will let $\mathcal{R}\cap A$ denote this $r_2$-division.

\section{Bicapacitated biconnectivity}\label{sec:capacitated}
Consider the forest of BC-trees of a graph. It may be viewed as a \emph{bicapacitated graph}, where non-bridge blocks have capacity $2$ and bridge blocks and vertices have capacity $1$; then, vertices $u$ and $v$ in $G$ are biconnected exactly when there exists a flow of value $2$ from $u$ to $v$ in the forest of BC-trees of $G$. (Here, we disregard the capacity of the source and sink vertices of a flow.) %
We denote by \emph{bicapacitated biconnectivity} the query to the existence of such a flow. 

Recall that we want to be able to use the framework recursively: we want to build and maintain BC-trees for small graphs and stitch them, or rather, compressed versions of them together, thus obtaining a patchwork graph. So, to make ends meet, we need to extend our definitions so that they can handle a bicapacitated input graph corresponding to the BC-trees of an underlying region.

The resulting patchwork graphs are always bipartite, with vertices on one side all having capacity $1$, and vertices on the other side having capacity either $1$ or $2$. We will restrict our definition of bicapacitated graph to mean such graphs.

Now, we can introduce the problem of \emph{fully dynamic bicapacitated biconnectivity}, as that of facilitating bicapacitated biconnectivity queries between vertices in a bicapacitated graph as it undergoes insertions and deletions of edges. 
Note that (fully) dynamic bicapacitated biconnectivity has an easy reduction to (fully) dynamic biconnectivity, as used in the following Lemma and its proof:
\begin{lemma}\label{lem:fullydynamiccapacitated}
	Given a fully dynamic data structure for biconnectivity in general graphs using amortized $t_u(n)$ time per link or cut and (amortized/worst case) $t_q(n)$ per pairwise biconnectivity or nearest cutvertex query, there is a fully dynamic data structure for bicapacitated graphs that uses $O(t_u(2n))$ amortized time per edge insert/delete, and answers pairwise biconnectivity and nearest-cutvertex queries in (amortized/worst case) $O(t_q(2n))$ time.
\end{lemma}
\begin{proof}
	Replace each vertex of capacity $\geq 2$ with a set of two
        vertices, and link every copy of a vertex to every copy of its
        neighbors. This transformation gives an uncapacitated graph
        with at most $2n$ vertices on which we can use the uncapacitated
        data structure, and which can answer the queries we need in
        the correct time.  Each edge insert/delete in the bicapacitated
        graph is at most a constant number of inserts/deletes in the
        uncapacitated graph, so the asymptotic running time for
        updates is the same.
\end{proof}

\section{The patchwork graph}\label{sec:patchwork}

We are given an $r$-division $\mathcal{R}=\set{R_1,\ldots,R_k}$ of $G$, and we want to define a graph $G_{\mathcal{R}}$ of size $O(\abs{\boundary\mathcal{R}})$ that somehow captures all the biconnectivity relations that cross multiple regions.  We call the resulting $G_{\mathcal{R}}$ a \emph{patchwork graph}, because it is built by stitching together a suitable \emph{patch graph} for each region.

Our patch graph for each region is in turn based on the forest of BC-trees for the region. 
We \emph{compress} the forest of BC-trees of the region similarly to \cite{Galil:1999} as follows:

\begin{definition}
  Given a bicapacitated graph $G=(V,E)$, its forest of BC-trees $F=\BC(G)$, 
  and a subset of vertices $S\subseteq V$, define a node\footnote{Throughout the text we consistently denote vertices of $G$ by \emph{vertices}, and vertices of BC-trees and SPQR-trees as \emph{nodes}.} $x\in T$, where the tree $T$ is a component of $F$, to be
  \begin{itemize}
  \item \emph{$S$-critical} if $x=\meet_T(s_1,s_2,s_3)$ for some $s_1,s_2,s_3\in S$,
  \item \emph{$S$-disposable} if $x\not\in s_1\longleftrightarrow_T s_2$ for all $s_1,s_2\in S$, and
  \item \emph{$S$-contractible} otherwise.
  \end{itemize}
\end{definition}
\begin{definition}
  The \emph{compressed BC-forest} $\CBC(G,S)$ is the forest obtained from
  its forest of BC-trees by deleting all $S$-disposable nodes, and replacing each maximal
  path of $S$-contractible nodes that start and end in distinct
  blocks, with a single so-called \emph{pseudoblock} node with capacity
  $1$.
\end{definition}

\begin{definition}
  Given an $r$-division $\mathcal{R}=\set{R_1,\ldots,R_k}$ of a graph $G$, define the \emph{patchwork graph} $G_{\mathcal{R}}=\bigcup_{R\in\mathcal{R}}\CBC(R,\boundary R)$ to be the bicapacitated graph obtained by taking the (non-disjoint) union of compressed BC-forests $\CBC(R,\boundary R)$  for each region $R\in \mathcal{R}$.
\end{definition}

Any vertex of $G$ corresponds to a BC-vertex in $\BC(R)$ for some $R$. Some of these BC-vertices are either present or represented in $G_{\mathcal{R}}$. We thus want to define the representation of a vertex as the vertex in $G_{\mathcal{R}}$ representing its BC-node, when it exists:
\begin{definition}
  Given a patchwork graph $G_{\mathcal{R}}$ and a vertex $v$ of $G$, we define the representative $\CB(v)$ of $v$ as follows:
  \begin{itemize}
  \item If $v$ is a vertex of $G_{\mathcal{R}}$, then $\CB(v) = v$; else
  \item Let $R\in\mathcal{R}$ be the unique region containing $v$.  If $v$ is incident to a block in $\BC(R)$ that is not $S$-disposable, then $v$ is represented either by that block or the pseudoblock representing it.
  \item Otherwise, $v$ is not represented.
  \end{itemize}
  Overloading notation slightly, say that a vertex of the graph is critical, disposable, or contractible, if the BC-node representing it is. 
\end{definition}

\begin{observation}\label{obs:linearCBC}
There is a linear time algorithm for building the compressed BC-forest of a graph with respect to a given subset of vertices, and for finding the representatives of the vertices.
\end{observation}

\begin{lemma}\label{lem:preservebicon}
  Distinct vertices $u,v$ are biconnected in $G$ if and only if either
  \begin{enumerate}
  \item\label{it:internallybiconnected} At least one of $u,v$ is not a boundary vertex, and $u,v$ are biconnected in the at most one region $R$ containing both; or
  \item\label{it:pseudobiconnected} $\CB(u)=\CB(v)$ is %
  a pseudo-block whose unique neighbours are biconnected in $G_{\mathcal{R}}$; or
  \item\label{it:externallybiconnected} $\CB(u)$ and $\CB(v)$ are different and are biconnected in $G_{\mathcal{R}}$.
  \end{enumerate}
\end{lemma}

\begin{proof}
We will show that $u$ and $v$ are \emph{not} biconnected if and only
if all three conditions are false.
Assume $u$ and $v$ are not biconnected. Then they can clearly not be biconnected within some region $R$, so condition~\ref{it:internallybiconnected} is false.  If $\CB(u)=\CB(v)$, then this is a pseudo-block contracted from a chain containing the neighbors of $u$, $v$ in $\BC(R)$, and a cutpoint $c$ that separates them within $R$.  Consider the neighbors $u^\prime$ and $v^\prime$ to this pseudoblock in $G_{\mathcal{R}}$. If they were biconnected in $G_{\mathcal{R}}$ there would be a $u,v$ path in $G\setminus\set{c}$ contradicting our choice of $u,v$. Thus $u^\prime$ and $v^\prime$ are not biconnected in $G_{\mathcal{R}}$ and condition~\ref{it:pseudobiconnected} is false.
Finally, if $\CB(u)$ and $\CB(v)$ are different, then any cutvertex $c$ separating $u$ and $v$ in $G$ will either be a cutvertex in $G_{\mathcal{R}}$, or will be in a pseudoblock $\CB(c)$ with neighbors $u^\prime$ and $v^\prime$. Since $c$ is a cutvertex in $G$, $(u^\prime,\CB(c))$ and $(\CB(c),v^\prime)$ are bridges in $G_{\mathcal{R}}$, and $\CB(u)$ and $\CB(v)$ will be separated by at least one of them and are therefore not biconnected in $G_{\mathcal{R}}$ and condition~\ref{it:externallybiconnected} is false.

If, on the other hand, none of the three conditions are true, then, if $\CB(u)=\CB(v)$ is a pseudoblock whose neighbours in $G_{\mathcal{R}}$ are not biconnected, then any cutvertex separating $u$ from $v$ in their region also separates them in $G$. If $\CB(u)\neq \CB(v)$ are separable by some cutvertex $c$ in $G_{\mathcal{R}}$, then $c$ is also a cutvertex separating $u$ from $v$ in $G$, and hence they are not biconnected. If $\CB(u)\neq \CB(v)$ are the endpoints of a bridge in $G_{\mathcal{R}}$ then one of them must be a pseudoblock containing a cutvertex or a bridge in $G$ separating them.
\end{proof}

Lemma~\ref{lem:preservebicon} above almost enables us to transform a biconnectivity-query in $G$ into a biconnectivity-query in $G_{\mathcal{R}}$ and a biconnectivity inside a region $R$. However, in fact, item~\ref{it:internallybiconnected} is only directly useful when neither of the vertices belong to the boundary; when one is a boundary vertex we do not know which vertex in the region it corresponds to. Fortunately, when the non-boundary vertex is represented, we may query biconnectivity in $G_{\mathcal{R}}$ to obtain the answer. To handle disposable vertices, we introduce the notion of the \emph{nearest represented vertex}:

\begin{definition}
When an $S$-disposable vertex $v$ is connected to at least one boundary vertex $b$, it knows 
its \emph{nearest represented vertex} $\nr(v) $ which is the first non-disposable node in the BC-tree of the region on the path from $b$ to $v$ (note that this node is one unique cutvertex). 
When an $S$-disposable vertex $v$ is not connected to the boundary, it has $\nr(v)=\nil$. 
\end{definition}

Note also that item~\ref{it:externallybiconnected} requires the pseudo-block to know its exactly two neighbours. 

\begin{lemma}\label{lem:preservebiconII}Vertices $u$ and $v$ are biconnected if and only if either
	\begin{itemize}
		\item $u$ and $v$ are non-boundary vertices of the same region and are biconnected in the region, or
		\item $u$ is a non-boundary vertex that is biconnected in its region $R$ with $\nr(u)$ and $\nr(u)=v$, or
		\item $\CB(u) = \CB(v)$ is a pseudo-block and its neighbours are biconnected in $G_{\mathcal{R}}$, or
		\item $\CB(u)\neq \CB(v)$ are biconnected in $G_{\mathcal{R}}$
	\end{itemize}
\end{lemma}
\begin{proof}
	Follows from Lemma~\ref{lem:preservebicon} by expanding item~\ref{it:internallybiconnected} into the two cases of whether both or only one vertex is non-boundary.
\end{proof}

Note that patchwork graphs are well-behaved and respect sub-divisions of $r$-divisions in the following sense:
\begin{lemma}\label{lem:samecbc}
If $S\subseteq \partial \mathcal{R}$, then $\CBC(G,S) = \CBC(G_{\mathcal{R}},S)$
\end{lemma}

\begin{proof}
It is enough to show a correspondence between the critical, disposable, and contractible BC-nodes. 

Consider an $S$-critical BC-node $x$ of $G$. It may overlap with several regions. However, in each region, each vertex of $x$ lies on some $r_1 \longleftrightarrow r_2$ path for $r_1,r_2\in\partial\mathcal{R}$, so they are never disposable. But then, since $S\subseteq \partial \mathcal{R}$, $x$ is also $\partial \mathcal{R}$-critical, and thus, present in $G_{\mathcal{R}}$. Clearly, once the block is present in $G_{\mathcal{R}}$, it is also $S$-critical in $G_{\mathcal{R}}$.

If a BC-node of $G$ is $S$-disposable, we only need to observe that its $\partial\mathcal{R}$-contractible and $\partial\mathcal{R}$-critical parts, for each path $r_1 \longleftrightarrow r_2$ they lie on, at most one endpoint is not $S$-disposable. 

Finally, if a BC-node $x$ of $G$ is $S$-contractible, then it lies on some path $s_1\longleftrightarrow s_2$, which $\partial \mathcal{R}$ cuts up into subpaths $r_1\longleftrightarrow r_2\longleftrightarrow r_3 \longleftrightarrow \ldots$ in (not necessarily different) regions $R_1, R_2, R_3, \ldots$. But then, all parts of $x$ are preserved as either $\partial\mathcal{R}$-critical or $\partial\mathcal{R}$-contractible $\CBC(R_i)$-vertices, and thus, survive in $\CBC(G_{\mathcal{R}},S)$. On the other hand, if a vertex in $R_i$ does not belong in $x$, then it does not lie on any of the paths $r_j\longleftrightarrow r_{j+1}$, and can thus not be represented by a vertex or a pseudo-block on that path.
\end{proof}

The same lines of thought can be used to make the following observation about how nested $r$-divisions behave with respect to patchwork graphs:
\begin{observation}
If $\boundary\mathcal{R}_1\subseteq \boundary\mathcal{R}_2$, then $G_{\mathcal{R}_1} = (G_{\mathcal{R}_2})_{\mathcal{R}_1}$. 
\end{observation}

\section{Decremental Biconnectivity in Patchwork Graphs}\label{sec:decremental}
Given the forest of BC-trees for (the patchwork graph associated with) each region of $G$ in an $r$-division $\mathcal{R}$, we want to explicitly maintain $G_{\mathcal{R}}$ and $\BC(G_{\mathcal{R}})$.

Let $R^\prime$ be a bicapacitated graph associated with region $R$, and suppose that $\CBC(R^\prime,\boundary R)=\CBC(R,\boundary R)$.   We will arrange things so either $R^\prime=R$ (with all vertices having capacity $1$), or $R^\prime=R_{\mathcal{R}^\prime}$ for some $r$-division $\mathcal{R}^\prime$ of $R$ with $\boundary R\subseteq\boundary\mathcal{R}^\prime$, so the equality follows from Lemma~\ref{lem:samecbc}.

We will maintain a fully dynamic biconnectivity structure for $R^\prime$ with amortized time $t(n)\in O(\poly(\log n))$ per operation, e.g. using~\cite{Holm:2001}\footnote{A faster algorithm here would just make more pairs of $r$-divisions suitable.}.
We use this structure to explicitly maintain $\BC(R^\prime)$ under the
following operations:
\begin{description}
\item[path deletion] --- given a path between two vertices of capacity $1$, whose internal vertices all have degree $2$, deletes all edges and internal vertices on the path.
\item[block split] --- given a vertex $u$ of capacity $2$, and an adjacent vertex $v$ of capacity $1$, split $u$ into two vertices $u_1,u_2$ of capacity $2$ connected by a path with $2$ edges via $v$, with $u_1,u_2$ partitioning the remaining neighbors of $u$.
\item[pseudoblock contraction] --- given a path of $3$ vertices, all having degree $2$ and the middle having capacity $1$, contract the path to a single vertex with capacity $1$.
\end{description}
The point is that if one of these operations is applied to $R^\prime$, then the change to $\BC(R^\prime)$ and $\CBC(R^\prime,\boundary R)$ can also be described by a sequence of these operations.\footnote{By \emph{explicit maintenance} is meant that each rooted BC-tree is maintained such that finding the parent of a vertex or block takes constant time.}

\begin{lemma}\label{lem:explicitBC}
  There is a data structure that explicitly maintains $\BC(R^\prime)$ 
  that can be initialized, and support any sequence of $O(\abs{R^\prime})$ path deletions, block splits, and pseudoblock contractions, in $O(\abs{R^\prime}t(n^\prime)\log n^\prime)$ total time, where $n^\prime$ is the number of vertices in $R^\prime$.
\end{lemma}
\begin{proof}
  Use the data structure from Lemma~\ref{lem:fullydynamiccapacitated} as a subroutine. Start by inserting all the edges. Each pseudoblock contraction can be simulated using a constant number of edge insertions or deletions. The total number of edges participating in path deletions is upper bounded by $O(n^\prime)$.
  Each block split either takes only a constant number of edge insertions or deletions, or makes a non-trivial partition of the adjacent edges. In the latter case, we still do a constant number of edge insertions and deletions, followed by one edge move (deletion and insertion) for each edge that ends up in a non-largest set in the partition.  Each edge is moved in this way 
  $O(\log n^\prime)$ times, so the total number of update operations done on the fully dynamic structure is %
  $O(\abs{R^\prime}\log n^\prime)$.  Once an update to $R^\prime$ has been simulated in the fully dynamic structure, we can use queries in that structure to find any new cutvertices that we need to update $\BC(R^\prime)$.
If the update in $R^\prime$ was a path deletion, then the corresponding update to $\BC(R^\prime)$ is either a path deletion, or a sequence of block splits. Each of these block splits can be found using the cutvertices given by the fully dynamic structure: Do a parallel search from both endpoints, and use the nearest cutvertex-query from Lemma~\ref{lem:fullydynamiccapacitated} to guide the search and to know when a whole block has been found. 
If the update in $R^\prime$ was a block split, this will either do nothing in $\BC(R^\prime)$ or cause a single block split.  If the update in $R^\prime$ is a pseudoblock contraction, the corresponding update to $\BC(R^\prime)$ is at most one edge deletion (because the leaf corresponding to the cutvertex disappears), at most one pseudoblock contraction (corresponding to the same pseudoblock contraction), or nothing happens (because the pseudoblocks were disposable).
\end{proof}

The point is that we will be using this with $\abs{R^\prime}=O(n/(t(n)\log n))$, where $n$ is the number of vertices in $R$, which means the total time used on $R$ is $O(n)\subseteq O(\abs{R})$.

\begin{lemma} The data structure above (Lemma~\ref{lem:explicitBC}) can be extended to handle also the explicit changes to $\CBC(R^\prime,\boundary R)$%
.  Any sequence of $O(\abs{R^\prime})$ updates to $R^\prime$ cause %
$O(\abs{\boundary R})$ updates in $\CBC(R^\prime, \boundary R)$.\label{lem:explicitCBC}
\end{lemma}
\begin{proof}
	For each change to $\BC(R^\prime)$, we can update $\CBC(R^\prime,\boundary R)$ accordingly. This essentially consists of replaying the same change as in $\BC(R^\prime)$, followed by at most two pseudoblock contractions; at most one in each end of the path it possibly unfolds to.  Note however, that some operations will end up having no effect on the structure of $\CBC(R^\prime,\boundary R)$. For example a trivial block split followed by a pseudoblock contraction will change only which cutvertex separates the pseudoblock from the block.  In this case, rather than doing a split and a contract, we simply update the identity of the cutvertex.  With this optimization, the total number of block splits is upper bounded by $O(\abs{\boundary R})$, and so is the number of edges and hence the number of possible path deletions and pseudoblock contractions.
\end{proof}

It immediately follows that we are able to efficiently maintain the patchwork graph, by combining the lemma above with the definition of the patchwork graph,  $G_{\mathcal{R}}=\bigcup_{R\in\mathcal{R}}\CBC(R,\boundary R).$
\begin{lemma}\label{lem:patchwork}
  Given a graph $G$, and a strict $(r,s)$-division $\mathcal{R}$ of $G$, if we can explicitly maintain $\CBC(R,\boundary R)$ for each $R\in\mathcal{R}$ in amortized constant time per update after $O(\abs{R})$ preprocessing, then we can explicitly maintain $G_{\mathcal{R}}$ in amortized constant time per update after $O(\abs{G})$ preprocessing.  Furthermore, any sequence of $O(\abs{G})$ updates in $G$ cause %
  $O(\abs{G_\mathcal{R}})$ updates in $G_\mathcal{R}$.
\end{lemma}
\begin{proof}
  Let $G$ have $n$ vertices and $m$ edges.  The first part follows
  trivially from $\sum_{R\in\mathcal{R}}\abs{R}\in
  O(n/r)O(r)+m=O(n+m)$.  Each block split in $G_{\mathcal{R}}$ either
  reduces the degree of some block, or adds a pseudoblock.  Since we
  do not add another pseudoblock when there already is one in a given
  direction, the maximum total number of splits an initial block
  vertex $v$ can cause is $O(d(v))$.  Thus the maximum number of
  splits is $\sum_{v\in G_{\mathcal{R}}}O(d(v))=O(\abs{G_{\mathcal{R}}})$,
  and so is the maximum number of edges and hence the number of
  possible path deletions and pseudoblock contractions.
\end{proof}

In order to use Lemma~\ref{lem:preservebiconII} to answer biconnected queries, we need to store some auxiliary information: for each pseudoblock, store its neighbours, and for each disposable vertex, store its nearest represented vertex. Thus, these need to be updated as the graph undergoes dynamic updates.

\begin{description}
	\item[path deletion] When a path from $x$ to $y$ is deleted, all vertices represented by internal nodes on the path become disposable. For each such vertex $v$, its nearest represented vertex becomes either $x$ or $y$. Furthermore, each vertex $u$ who had $v$ as its nearest represented vertex, now changes its nearest represented vertex to $\nr(u)=\nr(v)$. 
	In other words, the set of vertices having $x$ (or $y$) as a representative, is now the union of: vertices on the path, vertices represented by blocks or pseudo-blocks on the path, and the sets that these vertices used to represent.
	Note that these sets of vertices that have the same representative can be maintained via union find in $O(n \log n)$ total merge-time and $O(1)$ worst-case find-time, using the weighted quick-find algorithm (usually attributed to McIlroy and Morris, see~\cite{Aho:1974:DAC:578775}).
	As for the endpoints $x$ and $y$, they may change status from being represented by themselves to being represented by a block or pseudoblock.
	
	\item[block split] Note that a block is never the neighbour of a pseudoblock, nor is it the nearest represented \emph{vertex}, so block splits do not give cause to changes in neighbours and representatives.
	
	\item[pseudoblock contraction] does not give rise to changes in the nearest represented vertex - the vertices that were previously represented by a node that is involved in the contraction, are still represented, but now they are represented by the resulting pseudoblock. The set of vertices represented by the resulting pseudoblock is the union of vertices represented by nodes along the contracted path, again, this is done via union-find. Finally, the resulting pseudoblock is updated to remember its two neighbours.
\end{description}

We are now ready to prove:%
\begin{theorem}[first part of Theorem~\ref{thm:main}]\label{thm:partI}
	There exists a data structure that given a graph $G$ with $n$ vertices and $m$ edges, and given a suitable pair of $r$-divisions,  preprocesses $G$ in $O(m+n)$ time and handles any series of edge-deletions in $O(m)$ total time while answering queries to pairwise biconnectivity in $O(1)$ time.
\end{theorem}

\begin{proof}%
Given a fine $r$-division $\mathcal{R}$ of $G$, build the forest of BC-trees and compressed BC-forest for each region, and build the patchwork graph $G_{\mathcal{R}}$. Given the coarse division $\mathcal{A}$, and given the patchwork graph for $\mathcal{R}$, build the forest of BC-trees and the compressed BC-forest for each region of the patchwork graph, and build the patchwork graph $G_{\mathcal{A}}$. Finally, build the forest of BC-trees for $G_{\mathcal{A}}$. 
The construction time is linear, due to \cite{Tarjan72} and  Observation~\ref{obs:linearCBC}.

Deletions are handled bottom up: updating the regions of $\mathcal{R}$, then the regions of $G_{\mathcal{R}}$ induced by $\mathcal{A}$, and then $G_{\mathcal{A}}$. The total time for deletions is linear, due to Lemmata~\ref{lem:explicitBC}, \ref{lem:explicitCBC}, and~\ref{lem:patchwork}.

In detail:  Since $r_2=O(\poly(\log\log n))$, we can afford to precompute and store a table of all simple graphs on $r_2$ vertices with $s_2$ boundary vertices, and how their BC-trees and compressed BC-trees change under any possible edge deletion.  Using such a table, the region $R$ in $\mathcal{R}$ containing the deleted edge can be updated in constant time.

The updates to $\CBC(R,\boundary R)$ may cause some updates to the
patchwork graph $A_{\mathcal{R}\cap A}$ for the region
$A\in\mathcal{A}$ containing the deleted edge.  By
Lemma~\ref{lem:patchwork}, we can find these in amortized constant
time per edge deletion in $A$, and there are at most
$\abs{A_{\mathcal{R}\cap A}}$ of them.  Since
\begin{align*}
  \abs{A_{\mathcal{R}\cap
      A}}\leq\sum_{R\in\mathcal{R}\cap A}\abs{\boundary R}\leq s_2\cdot
  O\Paren*{\frac{r_1}{r_2}} \in O\Paren*{r_1\frac{s_2}{r_2}}
  \quad\text{and}\quad
  \frac{r_2}{s_2}=\Omega(t(r_1)\log r_1)
\end{align*}
we have
$\abs{A_{\mathcal{R}\cap A}}\in O(\frac{r_1}{t(r_1)\log r_1})$. By
Lemma~\ref{lem:explicitBC} and~\ref{lem:explicitCBC} we can therefore
explicitly maintain $\BC(A_{\mathcal{R}\cap A})$ and
$\CBC(A_{\mathcal{R}\cap A},\boundary A)$ in amortized constant time
per edge deletion in $A$.

The updates to $\CBC(A_{\mathcal{R}\cap A},\boundary A)$ again trigger
some number of updates to $G_{\mathcal{R}}$.  By
Lemma~\ref{lem:patchwork}, we can find these in amortized constant
time per edge deletion in $G$, and there are at most
$\abs{G_{\mathcal{R}}}$ of them.  Since
\begin{align*}
  \abs{G_{\mathcal{R}}}\leq
  \sum_{A\in \mathcal{A}}\abs{\boundary A}\leq s_1\cdot
  O\Paren*{\frac{n}{r_1}}\in O\Paren*{n\frac{s_1}{r_1}}
  \quad\text{and}\quad
  \frac{r_1}{s_1}=\Omega(t(n)\log n)
\end{align*}
we have
$\abs{G_{\mathcal{R}}}\in O(\frac{n}{t(n)\log n})$.  By
Lemma~\ref{lem:explicitBC} we can therefore explicitly maintain
$\BC(G_{\mathcal{R}})$ in amortized constant time per edge deletion in
$G$.

To handle biconnected-queries, perform the $O(1)$ queries indicated by Lemma~\ref{lem:preservebiconII}.
\end{proof}

\section{Nearest cutvertex in $O(1)$ worst-case time}\label{sec:nearest-cutvertex}
We have now shown how we handle queries to biconnectivity in a decremental graph subject to deletions. To answer nearest cutvertex queries, we need more structure. We need to augment our explicit representation of the dynamic BC-tree subject to block-splits so that it answers nearest cutvertex queries (subsection~\ref{sub:navigating}), and we need to show that we need only a constant number of queries in the patchwork graph together with a constant number of queries in regions, to answer nearest-cutvertex in the graph. 
\subsection{Navigating a dynamic BC-tree}\label{sub:navigating}
If the vertices $u$ and $v$ are connected but not biconnected, and we have a BC-tree over the component containing them, the nearest cutvertex to $u$ will be the second internal node on the unique BC-tree-path from $u$ to $v$. So, in order to answer nearest cutvertex queries, it is enough to answer first-on-path queries on a tree (since second-on-path can be found using two first-on-path queries).

\begin{lemma}\label{lem:repBC}
  There is a data structure for representing a dynamic BC-forest that
  can be initialized on a forest with $n$ nodes and
  support any sequence of $O(n)$ path-deletions, block-splits, and
  pseudoblock-contractions, in $O(n\log n)$ total time, while
  answering $\connected$ and $\firstonpath$ queries in worst case constant time.
\end{lemma}
\begin{proof}
  We use the data structures from
  Lemma~\ref{lem:tree-split-contract-ca-slow} and
  Lemma~\ref{lem:tree-connectivity} in the appendix as a base.  First,
  observe that we can combine these into a single structure,
  supporting both $\splitnode$, $\contract$, and $\deleteedge$
  operations and both $\firstonpath$ and $\connected$ queries.  This
  is because $\firstonpath(u,v)$ is only valid when $u$ and $v$ are
  connected, and the results of valid queries are therefore not
  affected by edge deletions. So we can maintain the two structures in
  parallel, and simply ignore deletions in the $\firstonpath$
  structure, and let each structure answer the query it is designed
  for.

  Second, observe that:
  \begin{itemize}
  \item Each path-deletion can be simulated using contractions and an
    edge deletion.
  \item Each block-split can be implemented as two node splits and an
    edge contraction.
  \item Each pseudoblock-contraction can be implemented as two edge
    contractions.
  \end{itemize}
  And note that if we color each vertex black and each block white
  (with pseudoblocks being either black or white depending on their
  history), then these operations respect the color requirements for
  our data structures.

  Since we do only $O(n)$ operations, and we start with $n$ black
  nodes, the total time for all updates is $O(n\log n)$.
\end{proof}

It follows as a corollary that we can answer nearest cutvertex queries given an explicit representation of the forest of BC-trees:

\begin{corollary}\label{cor:dynBC}
	Given a dynamic BC-tree over a connected $n$-vertex graph, we can answer biconnected and nearest cutvertex queries in $O(1)$ time, spending an additional $O(n\log n)$ time on any sequence of updates.
\end{corollary}
\begin{proof}
	Given a pair of connected and different vertices $u,v$, let $w$ be the second-on-path vertex found by querying $\firstonpath(\firstonpath(u,v),v)$. If $w=v$, the vertices are biconnected. Otherwise, $w$ is the nearest cutvertex separating $u$ from $v$.
\end{proof} 

\subsection{The patchwork graph}
In the following, recall that each disposable vertex $v$ knows its nearest represented vertex $\nr(v)$, and each pseudoblock knows its two neighbours.

\begin{lemma}\label{lem:nearest}
	If $u,v$ are connected and not biconnected, then the nearest cutvertex separating $u$ from $v$ can be determined by at most one nearest cutvertex-query and at most one biconnected query in $G_{\mathcal{R}}$ followed by at most one nearest cutvertex-query and at most one biconnected query within a region.
\end{lemma}

\begin{proof}
	If $u$ and $v$ are not both non-boundary vertices, and they are connected within the region $R$ containing both, then the nearest cutvertex within $R$ is the nearest cutvertex in $G$. 
			
	Otherwise, if $u$ is disposable, then it knows its nearest represented vertex $\nr(u)$. If $u$ and $\nr(u)$ are biconnected, then $\nr(u)$ is the nearest cutvertex separating $u$ from $v$ in $G$. Otherwise, the nearest cutvertex separating $u$ from $\nr(u)$ in their region $R$ is also the nearest cutvertex separating $u$ from $v$ in $G$.
	
	If $u$ is represented but $v$ is not represented, then $v$ knows its closest represented vertex $\nr(v)$ within its region. If $\CB(\nr(v))$ is biconnected with $\CB(u)$, then $\nr(v)$ is the answer, otherwise, $\nr(v)$ is used in place of $v$ in the following. 

	For the remaining cases, $u$ and $v$ are both represented, and their representatives are different. If the nearest cutvertex query between $\CB(u)$ and $\CB(v)$ in $G_{\mathcal{R}}$ returns a neighbour $b$ of the pseudo-block that either is $\CB(u)$ or is a neighbour of $\CB(u)$, then querying nearest cutvertex between $u$ and $b$ in the region of the pseudo-block will return the nearest cutvertex between $u$ and $v$ in $G$. Note here, that a pseudo-block is only present in one region, and even if $u$ is a boundary vertex that appears in several regions, the pseudo-block knows the identity of both its endpoints within the region.
			
	Finally, in all other cases, the nearest cutvertex separating $\CB(u)$ and $\CB(v)$ in $G_{\mathcal{R}}$ is the nearest cutvertex separating $u$ and $v$ in $G$.	
\end{proof}

\begin{theorem}[Second part of Theorem~\ref{thm:main}]\label{thm:partII}
	The data structure in Theorem~\ref{thm:partI} can be augmented to 
	support queries to nearest cutvertex in $O(1)$ worst-case time, while handling any series of edge-deletions in $O(n+m)$ total time.
\end{theorem}

\begin{proof}%
For the patchwork graph of $G$ and for the patchwork graphs of each region, maintain their dynamic BC-forest as indicated in Lemma~\ref{lem:repBC}. For the regions of the fine $r$-division, that is, those of $\polylog\log$-size, maintain an explicit table over the answer to nearest cutvertex queries.

Due to Corollary~\ref{cor:dynBC}, the maintenance of explicit forests of BC-trees over the patchwork graph of $G$ is done in $O(n'\log n')$ total time for $n'=O(n/\log n)$, thus, $O(n)$ total time, while handling intermixed nearest cutvertex-queries in $O(1)$ worst-case time. Same goes for the explicit maintenance of BC-forests of the patchwork graphs in the regions of the coarse $r$-division.

Finally, to handle nearest-cutvertex$(u,v)$-queries, perform the $O(1)$ queries indicated by Lemma~\ref{lem:nearest}: each of the $O(1)$ queries in the regions of the coarse $r$-division give rise to %
$O(1)$ look-ups in the regions of the fine $r$-division. Thus, the total query-time is constant.
\end{proof}

\section{Conclusion and implications.}
We have given a somewhat technical theorem stating that if a graph has suitable $r$-divisions, there is an efficient data structure for decremental biconnectivity. 
In this section, we show that for graphs with suitable $r$-divisions, $2$-edge connectivity reduces to biconnectivity. Thus, using our data structure as a blackbox, one obtains efficient data structures for decremental $2$-edge connectivity and bridge-finding.

For bounded genus graphs, we promised not only that they admit suitable $r$-divisions, but that such $r$-divisions can be computed in linear time in the size of the graph (regarding the size of the excluded minor as a constant). In this section, we state this as a theorem, although we defer its proof to an appendix. 
Furthermore, we address the natural problem of vertex deletions. 

\paragraph{Implications for connectivity and $2$-edge connectivity}
Our Theorem~\ref{thm:main} for decremental biconnectivity has immediate consequences for the questions of connectivity and $2$-edge connectivity, following from simple or relatively simple reductions.
\begin{corollary}\label{cor:connect}
There exists a data structure that, given a graph $G$ with $n$ vertices and $m$ edges, and given a suitable pair of $r$-divisions, preprocesses $G$ in $O(m+n)$ time and handles any series of edge-deletions in $O(m)$ total time while answering connectivity queries in $O(1)$ time.
\end{corollary}
\begin{proof}
Construct a graph $G'$ by adding to $G$ a dummy vertex $x$ connected to every other vertex. Add $x$ to the boundary sets of both of the nested $r$-divisions for $G$.  A pair of vertices are connected in $G$ if and only if they are biconnected in $G'$, the $r$-divisions are still a suitable pair, and thus, Theorem~\ref{thm:main} yields the result. 
\end{proof}
 
\begin{corollary}\label{cor:2edge}
There exists a data structure that, given a graph $G$ with $n$ vertices and $m$ edges, and given a suitable pair of $r$-divisions, preprocesses $G$ in $O(m+n)$ time and handles any series of edge-deletions in $O(m)$ total time while answering queries to pairwise $2$-edge connectivity and queries to nearest separating bridge in $O(1)$ time.
\end{corollary}
\begin{proof}
The reduction relies on an idea presented in Tazari and M\"uller-Hannemann~\cite{DBLP:journals/dam/TazariM09}. The basic idea is to perform a vertex-splitting in $G$, splitting a vertex $v$ of degree $d(v)$ into $d(v)$ vertices of degree $3$ arranged on a cycle. Unfortunately, since this splitting is also performed on boundary vertices of the $r$-division, the naive solution may increase the size of the boundary in the $r$-division. Thus, inspired by \cite{DBLP:journals/dam/TazariM09}, our algorithm makes the following preemptive measure: For each boundary vertex of an $r$-division, sort its incident edges into bundles according to the region of the other end vertex, and let all edges to other boundary vertices be bundled separately. Then, make a circular ordering of the edges, and ensure that all edges of a bundle are consecutive. Now, if the vertex splitting is performed according to this circular ordering, only two vertices for each bundle need to lie in the boundary. Thus, the total number of boundary vertices counted with multiplicity\footnote{That is, counting the distinct pairs of region and incident boundary vertex.} increases by at most a factor $2$. 

In the graph $G'$ resulting from the vertex splitting above, every edge still corresponds to a unique edge, but vertices are represented by cycles. An edge in $G$ is a bridge if and only if the corresponding edge in $G'$ is. The graph has max degree 3, and thus, biconnectivity and two-edge connectivity are equivalent, and any articulation point will be incident to a bridge. 
So, we can determine $2$-edge connectivity in $G$ by querying $2$-vertex connectivity in $G'$, and the first articulation point in $G'$ separating a pair of vertices will be incident to the first separating bridge between them. 
Since the size of $G'$ is still linear in the number of edges in $G$, and since the suitable pair of $r$-divisions of $G$ yields a suitable pair of $r$-divisions of $G'$, Theorem~\ref{thm:main} implies that we can determine $2$-edge connectivity and find a separating edge in worst-case constant time, with a total update time of $O(m)$.
\end{proof}

\paragraph{Implications for minor-free graphs}
We have stated our main theorem in terms of graphs with nested $r$-divisions with certain properties. Minor free graphs are a class of graphs that have these $r$-divisions, and in fact, we show that such $r$-divisions can be found in linear time, as stated in the following Theorem whose proof is deferred to Appendix A:

\begin{theorem}	
	\label{thm:minorfree-linear}
	Given a graph $G$ with $n$ vertices that does not have a $K_\ell$-minor, and any
	$t(n)\in O(\poly(\log n))$ we can compute a suitable pair of
	$r$-divisions in $O(n)$ time.	
\end{theorem}

Combined with Theorem~\ref{thm:main}, we thus obtain decremental biconnectivity with optimal worst-case query-time and amortized optimal update time.

\begin{corollary}
There exists a data structure that given a minor-free graph $G$ with $n$ vertices, preprocesses $G$ in $O(n)$ time and handles any series of edge- and vertex deletions in $O(n)$ total time while answering queries to pairwise connectivity,  $2$-edge connectivity, biconnectivity, nearest separating bridge in $O(1)$, and nearest separating cutvertex, in $O(1)$ time.	
\end{corollary}
\begin{proof}
Recall that if $G$ excludes a fixed minor, the number of edges in $G$ is $O(n)$. Theorem~\ref{thm:minorfree-linear} shows that $G$ admits a suitable pair of $r$-divisions that can be found in linear time. 
Thus, Theorem~\ref{thm:main}, Corollary~\ref{cor:connect}, and Corollary~\ref{cor:2edge} implies that we can preprocess $G$ to obtain data structures for connectivity, $2$-edge connectivity, and biconnectivity. 
Finally, notice that vertex-deletions can be simulated by edge-deletions: To delete a vertex, simply delete all its incident edges.
\end{proof}

Since the total number of edges in a minor-free graph is $O(n)$, the data structure above has the optimal amortized update time for edge deletions and vertex deletions, both. The question of whether our data structure generally admits vertex deletions in $O(n)$ total time remains open.

\appendix

\section{Fast r-divisions of minor free graphs}\label{sec:r-div}
This section is dedicated to the proof of Theorem~\ref{thm:minorfree-linear}.
We rely on the following $2$ known results:

\begin{lemma}[{Reed and Wood~\cite[Lemma 2]{reed:hal-01184376}}]\label{lem:rw}
There is an algorithm with running time $O(2^{2\ell}n+m)$ that, given $k,\ell\in \mathbb{Z}^+$, and an $n$-vertex $m$-edge graph $G$, outputs a connected $H$-partition $\{H_v | v \in V(G)\}$ such that either: 
\begin{itemize}
\item $H$ has a $K_\ell$ model, or
\item $|H| < 2^{(\ell^2+\ell-1)}k^{-1}n = O(k^{-1}n)$ and, for all
  $v$, $|H_v| \leq 2k$
\end{itemize}
\end{lemma}

\begin{lemma}[{Tazari and M\"uller-Hannemann~\cite[Lemma
  3.4]{DBLP:journals/dam/TazariM09}}]\label{lem:mh} Replacing the
  planar separator in Frederickson's \emph{Divide}
  procedure~\cite{Frederickson:87} with the separator algorithm of
  Reed and Wood~\cite[Theorem 2, which is based on
    Lemma~\ref{lem:rw}]{reed:hal-01184376} causes the
  Divide$(G,S,r,\ell)$ procedure to work as follows (where $G$ is a
  graph with $n$ vertices and excludes $K_\ell$ as a minor and $c_1$
  and $c_2$ are constants depending only on $\ell$):
  \begin{itemize}
  \item it divides $G$ into at most
    $c_2(\abs{S}/r^{\frac{2}{3}}+\frac{n}{r})$ regions;
  \item each region has at most $r$ vertices;
  \item each region has at most $c_1r^{\frac{2}{3}}$ boundary
    vertices, where the vertices in $S$ also count as boundary;
  \item it takes time $O(n\log n)$.
  \end{itemize}
\end{lemma}

The separator algorithm used in Lemma~\ref{lem:mh} can be replaced
with the improved separator algorithm in
Wulff-Nilsen~\cite{Wulff-Nilsen11} to improve the constant factors in
the running time and the dependency on $\ell$ (at the cost of a slightly worse separator size), but
we still need to use the algorithm in Lemma~\ref{lem:rw} to get linear time.

\begin{lemma}\label{lem:minorfree-rdiv-S}
Given a graph $G$ that does not have a $K_\ell$-minor, and a subset
$S$ of its vertices, for any $r\in\Omega(\log n)$ we can in linear
time compute a partition of $G$ into $O(\abs{S}(\frac{\log
  n}{r})^{\frac{2}{3}}+\frac{n}{r})$ regions, each with at most $r$
vertices and $O(r^{\frac{2}{3}}\log^{\frac{1}{3}}n)$ boundary
vertices.
\end{lemma}
\begin{proof}
Choose $k\leq r/2, k\in\Theta(\log n)$, such that the linear time
algorithm from Lemma~\ref{lem:rw} outputs a $H$-partition with
$\abs{H}\in O(n/\log n)$ where $\abs{H_v}\leq 2k\leq r$ for every
$v$. Construct the graph $G'$ with $n'\in O(n/\log n)$ vertices by
contracting each $H_v$ to a single vertex, and let $S'$ be the vertices contracted from $H_v$ where $v\in S$.

Now set $\mathcal{R}' :=
\operatorname{Divide}(G',S',\ceil{\frac{r}{2k}},\ell)$ using
the modified Divide algorithm from Lemma~\ref{lem:mh}.  Since $n' \in
O(n/\log n)$ this takes linear time, and the result is a strict
$(\ceil{\frac{r}{2k}}, c_1\ceil{\frac{r}{2k}}^{\frac{2}{3}})$-division
of $G'$ with at most $c_2(\abs{S'}/\ceil{\frac{r}{2k}}^{\frac{2}{3}} +
n'/\ceil{\frac{r}{2k}})\leq c_2(\abs{S'}(\frac{2k}{r})^{\frac{2}{3}} +
n'\frac{2k}{r})\in O(\abs{S} (\frac{\log n}{r})^{\frac{2}{3}}+ n/r)$
regions.

Finally, construct $\mathcal{R}$ by replacing each vertex in each
region of $\mathcal{R}'$ with the subgraph $H_v$ that it was
contracted from.  For each edge that this would put in multiple
regions, just pick the first region that contains it.  This also takes
linear time.

By choice of $k$, $\ceil{\frac{r}{2k}}<\frac{r}{2k}+1\leq \frac{r}{k}$,
and thus each region has at most $\ceil{\frac{r}{2k}}k<r$ vertices and
at most $c_1\ceil{\frac{r}{2k}}^{\frac{2}{3}}k\leq
c_1r^{\frac{2}{3}}k^{\frac{1}{3}}\in
O(r^{\frac{2}{3}}\log^{\frac{1}{3}}n)$ boundary vertices.
\end{proof}

\begin{corollary}\label{cor:minorfree-rdiv}
  Given a graph $G$ that does not have a $K_\ell$-minor, for any
  $r\in\Omega(\log n)$ we can compute a strict
  $(r,O(r^{\frac{2}{3}}\log^{\frac{1}{3}}n))$-division in linear time\footnote{We believe this fact to be folklore, but failed to find a reference in the literature.}.
\end{corollary}
\begin{proof}
  Simply apply Lemma~\ref{lem:minorfree-rdiv-S} with $S=\emptyset$.
\end{proof}

\begin{proof}[Proof of Theorem~\ref{thm:minorfree-linear}]
  First use Lemma~\ref{lem:minorfree-rdiv-S} with
  $r=r_1\in\Theta(t^3(n)\log^4 n)$ and $S=\emptyset$ to compute the
  strict $(r_1,s_1)$-division $\mathcal{A}$ in linear time, where
  $s_1\in O(r_1^{\frac{2}{3}}\log^{\frac{1}{3}}n)$.  Observe that
  $\frac{r_1}{s_1}\in\Omega((\frac{r_1}{\log
    n})^{\frac{1}{3}})=\Omega(t(n)\log n)$ as required.

  Then for each $A\in\mathcal{A}$, use Lemma~\ref{lem:minorfree-rdiv-S}
  again with $r=r_2\in\Theta(t^3(r_1)\log^4 r_1)$ and $S=\boundary A$
  to compute a partition $A\cap\mathcal{R}$ of $A$ into
  $O(\abs{\boundary A}(\frac{\log\abs{A}}{r_2})^{\frac{2}{3}} +
  \frac{\abs{A}}{r_2} ) \subseteq O(s_1(\frac{\log r_1}{r_2})^{\frac{2}{3}} +
  \frac{r_1}{r_2} ) = O(\frac{r_1}{r_2})$ regions, each having at
  most $r_2$ vertices and at most $s_2\in
  O(r_2^{\frac{2}{3}}\log^{\frac{1}{3}}\abs{A})\subseteq
  O(r_2^{\frac{2}{3}}\log^{\frac{1}{3}}r_1)$ boundary vertices.  The
  union of these regions, $\mathcal{R}$, has at most
  $\sum_{A\in\mathcal{A}}O(\frac{r_1}{r_2})=O(\frac{n}{r_2})$ regions,
  and is thus a strict $(r_2,s_2)$-division of $G$. Furthermore,
  $\frac{r_2}{s_2}\in\Omega((\frac{r_2}{\log
    r_1})^{\frac{1}{3}})=\Omega(t(r_1)\log r_1)$ as required.

  Finally, since $t(n)\in O(\poly(\log n))$, we have $r_1,s_1\in
  O(\poly(\log n))$ and $r_2,s_2\in O(\poly(\log\log n))$.
\end{proof}

\section{Dynamic first-on-path with split and contract}\label{sec:first-on-path}

In this section, we will be considering a tree with two kinds of
nodes, called \emph{black} and \emph{white}, and rooted at some node
$r$. For any node $u$, let $N(u)$ denote the set of neighbors of $u$
(including $\parent(u)$ if $u\neq r$), and let $d(u)=\abs{N(u)}$.

Our goal is to support the following operations:
\begin{description}

\item[split$(u,M)$:] given a white node $u$ with $d(u)\geq 2$, and a
  subset $M\subset N(u)$ of its neighbors, $1\leq \abs{M}\leq
  \frac{1}{2}d(u)$, insert a new white node $v$ as child of $u$.  Then
  let $M'=M$ if $u=r$ or $\parent(u)\not\in M$, and $M'=N(u)\setminus
  M$ otherwise.  Finally make each node in $M'$ a child of $v$.

\item[contract$(e)$] given an edge $e=(c,p)$, make all children of
  $c$ children of $p$ and turn $p$ black.

\item[parent$(u)$:] given a node $u$, return its parent (or $\nil$ if
  $u$ is the root).  This operation should take worst case $O(1)$
  time.

\item[first-child$(u)$:] given a node $u$, return its first child (or
  $\nil$ if $u$ is a leaf).  The order of the children is chosen
  arbitrarily by the data structure, but stays fixed between splits.
  This operation should take worst case $O(1)$ time.

\item[next-sibling$(u)$:] given a node $u$, return its next sibling (or
  $\nil$ if $u$ is the last child of its parent).  The order of the
  children is chosen arbitrarily by the data structure, but stays
  fixed between splits.  This operation should take worst case $O(1)$
  time.

\item[ca$(u,v)$:] given nodes $u,v$, return the triplet $(a,u',v')$ of
  \emph{characteristic ancestors}
  (See~\cite{Gabow:1990:DSW:320176.320229}) where $a=\nca(u,v)$ is the
  nearest common ancestor to $u,v$, and $u'$ (resp $v'$) is either the
  first node on the path from $a$ to $u'$ ($v'$), or $a$ if $u=a$
  ($u=v$).  This operation should take worst case $O(1)$ time.

\item[first-on-path$(u,v)$:] given distinct nodes $u,v$ return the
  first node $w$ on the path from $u$ to $v$.  Note that given
  $p=\parent(u)$ and $(a,u',v')=\ca(u,v)$ this can trivially be
  computed in $O(1)$ time as follows: If $u=a$ then $w=v'$, otherwise
  $w=p$.

\end{description}
The goal is that, when starting on a star with $B$ black leaves and a
white center, any sequence of $s$ splits and (at most $B+s$) contracts
takes total $O(s+B\log B)$ time, while answering
intermixed queries in worst case constant time.

Note that there is a sequence of valid splits
$\splitnode(u_1,M_1),\ldots,\splitnode(u_s,M_s)$ with
$\sum_{i=1}^s\abs{M_i}=\Omega(s+B\log B)$, so a total time of
$O(s+B\log B)$ is in some sense the best we can hope for without
somehow compressing the input.

Observe also that starting from a star is not really a restriction,
since given such a data structure, any tree with $n$ nodes of which
$B$ are black can easily be constructed from such a star in $O(n+B\log
B)$ time.

As a starting point we will use the following simple tree data structure.

\begin{lemma}\label{lem:tree-split-contract-parent}
  There is a data structure for a dynamic rooted tree with $n$ nodes,
  that can be initialized in $O(n)$ time, supporting $\splitnode(u,M)$
  operations in worst case $O(\abs{M})$ time, $\contract(u,v)$
  operations in worst case $O(\min\set{d(u),d(v)})$ time, and
  $\parent(u)$, $\firstchild(u)$, and $\nextsibling(u)$ queries in
  worst case constant time,
\end{lemma}
\begin{proof}
  We need a level of indirection to be able to handle
  $\splitnode(u,M)$ where $\parent(u)\in M$.  Instead of a direct
  pointer from each child to its parent, we keep the list of children
  in a doubly-linked ring with a ``sentinel''. The parent and each
  child in the ring has a pointer to the sentinel, and only the
  sentinel has a pointer to the parent.
  This structure can clearly be initialized in $O(n)$ time.
  Finding the parent then requires following two pointers instead of
  one, but that is still constant.
  Creating a new node $v$, and (if $\parent(u)\in M$) making all the
  original children of $u$ children of $v$, and then making $v$ a
  child of $u$ only requires changing a constant number of
  pointers.
  Moving each of the $\abs{M}$ edges between $u$ and $v$ takes only a
  constant number of pointer changes per edge.  Thus the total number
  of pointer changes per split is $O(\abs{M})$.

  For $\contract(v,u)$, we move the children of the shorter child list
  to the larger list one at a time (since they need to point to the
  right sentinel), then make that child list the child list of
  $u$. This clearly takes $O(\min\set{d(u),d(v)})$ time.
\end{proof}

\begin{lemma}
  The total time for $s$ splits and at most $B+s$ contractions, when
  using a data structure that uses worst case $O(\abs{M}$ time for
  $\splitnode(u,M)$ and worst case $O(\min\set{d(u),d(v)})$ for
  $\contract(v,u)$, is $O(s+B\log B)$.
\end{lemma}
\begin{proof}
  Since in each $\splitnode(u,M)$, we require $1\leq \abs{M} \leq
  \frac{1}{2}d(u)$, in any sequence of splits
  $\splitnode(u_1,M_1),\ldots,\splitnode(u_s,M_s)$ each node appears
  at most $O(\log B)$ times in a set $M_i$ with $\abs{M_i}>1$, and at
  most $O(B)$ nodes can ever appear in such a set. Thus for any such
  sequence, $\sum_{i=1}^s\abs{M_i}\in O(s+B\log B)$.

  Similarly, each node can only contribute to the minimum degree in a
  contraction at most $O(\log B)$ times, and at most $O(B)$ nodes can
  contribute more than a constant number of times.
\end{proof}

Our final data structure is inspired by the similar structure by
Gabow~\cite{Gabow:1990:DSW:320176.320229}, supporting
$\CAaddleaf(p,c)$ and $\CAdeleteleaf(c)$ in amortized constant time
and $\ca(u,v)$ in worst case constant time, but not
$\splitnode(c,M)$. We also use a list-ordering data structure in a
similar way to the data structure by Cole and Hariharan
2002~\cite{ColeHariharan:2005:doi:10.1137/S0097539700370539}, which
supported $\CAaddleaf(p,c)$ and $\CAdeleteleaf(c)$ in worst case
constant time, and also supported \emph{edge subdivision} and
\emph{degree $2$ node contraction} in worst case constant time.
However, many details (e.g. the definition of the size of a subtree)
are different.

\begin{definition}
  Define the \emph{black weight} $b(u)$ of each node $u$ to be $0$ if
  $u$ is white, and the number of original black nodes that have have
  been contracted to form $u$ otherwise (this might be $0$ if two
  white nodes are contracted).  Now define the size of $T_u$ as
  $s(u)=b(T_u)=\sum_{v\in T_u}b(v)$.
\end{definition}
Note that with this definition, the size of the root is $B$ at all
times, and that $\splitnode$ and $\contract$ leaves the size of all
existing nodes unchanged.

\begin{lemma}\label{lem:tree-split-contract-parent-size}
  The data structure from Lemma~\ref{lem:tree-split-contract-parent} can be
  augmented to answer $b(u)$ and $s(u)$ queries in worst case constant time,
  without changing the asymptotic running time of the other operations.
\end{lemma}
\begin{proof}
  During initialization, we compute
  \begin{align*}
    b(u) &= [\text{$u$ is black}]
    &
    s(u) &= b(u) + \sum_{c\in\children(u)}s(c) 
  \end{align*}
  for each node $u$ bottom up, and
  simply store it in each node.  This clearly takes linear time.
  During $\splitnode(u,M)$ compute the values for the new node $v$ as
  \begin{align*}
    b(v) &= 0
    &
    s(v) &=
    \begin{cases}
      \sum_{c\in M}s(c) & \text{if $\parent(u)\not\in M$}
      \\
      s(u)-\sum_{c\in M\setminus\set{\parent(u)}}s(c) & \text{if $\parent(u)\in M$}
    \end{cases}
  \end{align*}
  This computation takes $O(\abs{M})$ time worst case, so the
  asymptotic worst case time for $\splitnode(u,M)$ is unchanged
  $O(\abs{M})$.
  Similarly, during $\contract(v,u)$, we compute the new value of
  $b(u)$ as $b(u)+b(v)$ and leave the value of $s(u)$ unchanged, which
  clearly leaves the asymptotic worst case time for $\contract(v,u)$
  unchanged.
\end{proof}

\begin{definition}
  Let $p$ be the parent of $c$.  If $s(c)>\frac{1}{2}s(p)$, then
  $(c,p)$ is called \emph{heavy} and $c$ is a \emph{heavy child}.
  Otherwise $(c,p)$ is \emph{light}, and $c$ is a \emph{light child}.
\end{definition}

\begin{lemma}\label{lem:tree-split-contract-parent-size-sorted}
  The data structure from
  Lemma~\ref{lem:tree-split-contract-parent-size} can be augmented
  such that the children are ordered by decreasing
  $\rank(c)=\floor{\log_2 s(c)}$, and such that the heavy child (if
  any) comes first, without changing the asymptotic running times of
  any of the operations.
\end{lemma}
\begin{proof}
  We use the dynamic integer set structure by P\v{a}tra\c{s}cu and
  Thorup~\cite{Patrascu:2014:DIS:2706700.2707461} as a black box.  On
  the AC$^0$ RAM with word size $w$, it supports $\LOinsert$,
  $\LOdelete$, $\LOpred$, and $\LOsucc$ operations on sets of size $s$
  in worst case $O(\log s/\log w)$ time per operation, using linear
  space.  For $s=w^{O(1)}$, this is worst case constant time.
  Furthermore, using $O(n)$ time and space for preprocessing, and
  assuming $w=\Theta(\log n)$, the non-standard AC$^0$ operations
  and/or multiplications used by the structure can be replaced by
  table lookups, making this structure usable even on the practical
  RAM.

  Since the maximum rank is $\floor{\log_2 n}$ and we assume
  $w=\Omega(\log n)$, we can use this to maintain any set of ranks in
  linear space and worst case constant time per operation.  In
  particular, we will maintain an index for each child list that
  points to the first child with each rank\footnote{In practice, and
    in particular for our application, the use of this structure is
    overkill.  A much simpler solution is to make the index an array
    of $\floor{\log_2 n}+1$ pointers in each child list with at least
    that many children and just not store the index for lists with
    fewer children.  This still uses only linear space, and
    initializing the structure can still be done in linear time by
    using an initial radix sort on $(\parent,\rank)$.  Split can be
    implemented in worst case constant time when $\abs{M}=1$ and
    otherwise in worst case $O(\abs{M}+\log B)$ time (using
    bucketsort), which does not affect the amortized run time of the
    final structure. Similarly, $\contract(v,u)$ can be implemented in
    worst case constant time if one of $u,v$ only has one child, and in worst
    case $O(\min\set{d(u),d(v)}+\log B)$ time otherwise, again without
    changing the amortized run time of the final structure.}.

  After initializing the data structure from
  Lemma~\ref{lem:tree-split-contract-parent-size}, we can sort each child
  list and build the corresponding index by first removing all the
  children and then reinserting each child one at a time, using the
  index to find the correct position to insert it.  This clearly takes
  linear time.

  Once each list is sorted by decreasing rank, the heavy child (if any)
  is among the first $2$ elements.  Suppose for contradiction that $c$
  is a heavy child of $u$, but is not among the first $2$ children of
  $u$ in decreasing rank order.  Then there exists children
  $c_1,c_2$, with (for $i=1,2$)
  \begin{align*}
    \log_2 s(c_i)
    \geq \rank(c_i)
    \geq \rank(c)
    > (\log_2 s(c))-1
    =\log_2\Paren*{\frac{1}{2}s(c)}
  \end{align*}
  and thus $s(c_i)>\frac{1}{2}s(c)$. But then
  $s(c_1)+s(c_2)+s(c)>2s(c)>s(u)$, which is impossible because the
  size of the parent is at least the sum of the sizes of its children.
  Thus, if there is a heavy child, it must be among the first $2$
  elements in the child list.  After each change to the child list, we
  can therefore find the heavy child (if any) and move it to the front
  of the list, without breaking the sort order, in worst case constant
  time.

  Finally, each split does only $O(\abs{M})$ insert and delete
  operations on child sets, and contract needs only
  $\min\set{d(u),d(v)}$ insert and delete operations on child sets,
  and using the index each of these can be done in worst case constant
  time while preserving the required sort order.
\end{proof}

\begin{definition}
  Deleting all the light edges from $T$ partitions the nodes into
  \emph{heavy paths}.  For each vertex $u$ let $\widehat{u}$ denote
  the heavy path it belongs to.  The node closest to the root of
  $\widehat{u}$ is called the \emph{apex node} and denoted
  $\apex(\widehat{u})$.
  Let the \emph{light tree} $\widehat{T}$ be the tree obtained from
  $T$ by contracting all the heavy paths, such that for each vertex $u$,
  $\widehat{u}$ corresponds to a node in $\widehat{T}$.
  Let the \emph{light depth} of $v$, denoted $\ell(v)$, be the number
  of light edges on the path from the root to $v$, or equivalently the
  depth of $\widehat{v}$ in $\widehat{T}$.
\end{definition}

\begin{lemma}
  Every node $u$ has light depth $\ell(u)\leq\floor{\log_2 \frac{B}{s(u)}}$.
\end{lemma}
\begin{proof}
  By definition, for each edge $(c,p)$ on the root path of $u$ we have
  $s(c)\leq s(p)$, and if $(c,p)$ is light we have
  $s(c)\leq\frac{1}{2}s(p)$, so $s(u)\leq 2^{-\ell(u)}B$, and thus
  $\ell(u)\leq\log_2\frac{B}{s(u)}$.  Furthermore, since $\ell(u)$ is
  an integer, we can strengthen this to
  $\ell(u)\leq\floor{\log_2\frac{B}{s(u)}}$.
\end{proof}

\begin{lemma}\label{lem:light-change-trees-split}
  After a $\splitnode(u,M)$ creates a new node $v$:
  \begin{enumerate}
  \item\label{it:light-depth-increased-split} If $(v,u)$ is light, then for
    each light child $c$ of $v$, every node in $T_c$ had its light
    depth increased by $1$ by the $\splitnode$.
  \item\label{it:light-depth-decreased-split} If $(v,u)$ is heavy, and $v$
    has a heavy child $c$ with $s(c)\leq\frac{1}{2}s(u)$, then every
    node in $T_c$ had its light depth decreased by $1$ by the
    $\splitnode$.
  \item\label{it:light-depth-unchanged-split} Otherwise the $\splitnode$ did
    not change the light depth of any node.
  \end{enumerate}
\end{lemma}
\begin{proof}
  First note that only nodes that are in $T_v$ after the split can
  have changed their light depth, since the size of every existing
  node is unchanged.  Further note that, as seen from such a node $w$,
  the only change made to the root path of $w$ by the split is the
  insertion of $v$ between $u$ and its child $c$ on the path.
  Now let's consider the cases:
  \begin{description}
  \item[case~\ref{it:light-depth-increased-split}] If $(c,v)$ and $(v,u)$
    are both light, then $(c,u)$ was light and thus $\ell(w)$ has
    increased by $1$.
  \item[case~\ref{it:light-depth-decreased-split}] If $(c,v)$ and $(v,u)$
    are both heavy, and $s(c)\leq\frac{1}{2}s(u)$, then $(c,u)$ was
    light and thus $\ell(w)$ has decreased by $1$.
  \item[case~\ref{it:light-depth-unchanged-split}] Otherwise either exactly
    one of $(c,v)$ and $(v,u)$ is light, in which case $(c,u)$ was
    light and $\ell(w)$ is unchanged, or both $(c,v)$ and $(v,u)$ are
    heavy and $s(c)>\frac{1}{2}s(u)$ so $(c,u)$ was also heavy and
    $\ell(w)$ is again unchanged.
  \end{description}
  Since this holds for all proper descendants of $T_v$, and no other
  nodes change their light depth, the result follows.
\end{proof}

\begin{lemma}\label{lem:light-change-trees-contract}
  After a $\contract(e)$ contracts edge $e=(v,u)$:
  \begin{enumerate}
  \item\label{it:light-depth-decreased-contract} If $(v,u)$ was light, then for
    each light child $c$ of $v$, every node in $T_c$ had its light
    depth decreased by $1$ by the $\contract$.
  \item\label{it:light-depth-increased-contract} If $(v,u)$ was heavy, and $v$
    had a heavy child $c$ with $s(c)\leq\frac{1}{2}s(u)$, then every
    node in $T_c$ had its light depth increased by $1$ by the
    $\contract$.
  \item\label{it:light-depth-unchanged-contract} Otherwise the $\contract$ did
    not change the light depth of any node.
  \end{enumerate}
\end{lemma}
\begin{proof}
  First note that only nodes that were in $T_v$ before the split can
  have changed their light depth, since the size of every existing
  node is unchanged.  Further note that, as seen from such a node $w$,
  the only change made to the root path of $w$ by the $\contract$ is the
  removal of $v$ between $u$ and its child $c$ on the path, and $u$
  turning black.
  Now let's consider the cases:
  \begin{description}
  \item[case~\ref{it:light-depth-decreased-contract}] If $(c,v)$ and $(v,u)$
    were both light, then $(c,u)$ is light and thus $\ell(w)$ has
    decreased by $1$.
  \item[case~\ref{it:light-depth-increased-contract}] If $(c,v)$ and $(v,u)$
    were both heavy, and $s(c)\leq\frac{1}{2}s(u)$, then $(c,u)$ is
    light and thus $\ell(w)$ has decreased by $1$.
  \item[case~\ref{it:light-depth-unchanged-contract}] Otherwise either exactly
    one of $(c,v)$ and $(v,u)$ was light, in which case $(c,u)$ is
    light and $\ell(w)$ is unchanged, or both $(c,v)$ and $(v,u)$ were
    heavy and $s(c)>\frac{1}{2}s(u)$ so $(c,u)$ is also heavy and
    $\ell(w)$ is again unchanged.
  \end{description}
  Since this holds for all proper descendants of $T_v$, and no other
  nodes change their light depth, the result follows.
\end{proof}

\begin{lemma}\label{lem:few-light-changes-split-contract}
  The light depth of $u$ changes at most
  $\max\set{0,6\floor{\log_2\frac{B}{s(u)}}-1}$ times during any
  sequence of splits and contractions.
\end{lemma}
\begin{proof}
  Let $\pi(u)$ denote the root path of $u$, let
  \begin{align*}
    H(u) &:=\set*{(c,p)\in\pi(u)\cond (c,p) \text{ is heavy}}
    \\
    L_b(u) &:=\set*{(c,p)\in\pi(u)\cond (c,p) \text{ is light and $p$ is black}}
    \\
    L_w(u) &:=\set*{(c,p)\in\pi(u)\cond (c,p) \text{ is light and $p$ is white}}
  \end{align*}
  and define the potential of $u$ as
  \begin{align*}
    \Phi(u)
    &:=
    \abs{L_b(u)}
    +
    2\floor*{
      \sum_{(c,p)\in H(u)} \log_2\frac{s(p)}{s(c)}
    }
    +
    \sum_{(c,p)\in L_w(u)} \Paren*{6\floor*{\log_2\frac{s(p)}{s(c)}}-1}
  \end{align*}
  Observe that $0\leq \Phi(u)\leq
  \max\set{0,6\floor{\log_2\frac{B}{s(u)}}-1}$, and that no split or
  contract can increase $\Phi(u)$.  Finally note that (as seen from
  $u$):
  \begin{itemize}
  \item each split that changes $\ell(u)$ must replace a light child
    with a white parent on the root path either by two light edges
    with white parents (so the last term drops by at least $1$) or by
    two heavy edges (so the last term drops by $5$ and the middle term
    increases by at most $4$).

  \item each contract that changes $\ell(u)$ either contracts a light
    edge (decreasing either the first term by $1$ or the last by at
    least $5$), or contracts a heavy edge $(v,u)$ and makes the heavy
    child $c$ of $v$ light and makes $u$ black (decreasing the middle
    term by at least $2$ and increasing the first by $1$).
  \end{itemize}
  Thus $\Phi(u)$ decreases by at least one for each $\splitnode$ or
  $\contract$ that changes $\ell(u)$.
\end{proof}

\begin{lemma}\label{lem:tree-split-contract-ca-slow}
  There is a data structure for dynamic trees that, when initialized
  on a star with a white center and $B$ black leaves, can be
  initialized and handle any sequence of $s$ splits of white nodes and
  (at most $B+s$) contractions in $O((B+s)\log B)$ total time, while
  answering intermixed $\parent$, $\firstchild$, $\nextsibling$,
  $\ca$, and $\firstonpath$ queries in worst case constant time.
\end{lemma}
\begin{proof}
  We will use the tree from
  Lemma~\ref{lem:tree-split-contract-parent-size-sorted} as basis, and
  extend it using two existing data structures as black boxes.  For
  each heavy path $\widehat{x}$ we keep track of its apex node,
  $\apex(\widehat{x})$, and maintain a \emph{list-ordering} data
  structure (see
  e.g.~\cite{Dietz:1987:TAM:28395.28434,10.1007/3-540-45749-6_17}),
  supporting \emph{insert$_{\widehat{x}}(u,v)$} and
  \emph{delete$_{\widehat{x}}(u)$} operations in amortized constant
  time, and \emph{order$_{\widehat{x}}(u,v)$},
  \emph{succ$_{\widehat{x}}(a)$}, and \emph{pred$_{\widehat{x}}(a)$}
  queries in worst case constant time.  And for the light tree we
  maintain a \emph{characteristic-ancestors} data structure (see
  e.g.~\cite{Gabow:1990:DSW:320176.320229}) supporting
  \emph{add-leaf$(p,c)$} and \emph{delete-leaf$(c)$} in amortized
  constant time, and $\ca(\widehat{u},\widehat{v})$ queries in worst
  case constant time.

  Given the tree from Lemma~\ref{lem:tree-split-contract-parent-size-sorted},
  we can initialize the remaining parts of the structure as follows:
  Every node starts uninitialized.  For each node $u$ in depth-first
  preorder: If $u$ is the root or $p=\parent(u)$ and $(u,p)$ is light,
  create a new node $\widehat{u}$ and (if $u$ is not the root) add it
  to the light tree using $\CAaddleaf(\widehat{p},\widehat{c})$.  Also
  create a new empty list-order structure for $\widehat{u}$ and store
  $\apex(\widehat{u})=u$ in $\widehat{u}$.  Otherwise $(u,p)$ is
  heavy, so $\widehat{u}=\widehat{p}$, and we just add $u$ to the
  list-ordering structure of $\widehat{u}$ using
  $\LOinsert_{\widehat{u}}(p,u)$.  Finally store $\widehat{u}$ in $u$
  to mark it as initialized.  This takes $O(1)$ time per node, which
  is $O(n)$ because no adjacent nodes have degree $2$.

  Computing $(a,u',v')=\ca(u,v)$ in $T$ can then be done as follows:
  First compute
  $(\widehat{a},\widehat{u}',\widehat{v}')=\ca(\widehat{u},\widehat{v})$.
  Then let $u''=\apex(\widehat{u}')$ and $u'''=\parent(u'')$ if
  $\widehat{a}\neq\widehat{u}'$ and $u''=u'''=u$ otherwise, and
  compute $v''$ and $v'''$ symmetrically.  Now if $u'''=v'''$ set
  $(a,u',v')=(u''',u'',v'')$; otherwise if $\LOorder_{\widehat{a}}(u''',v''')$ set
  $(a,u',v')=(u''',u'',\LOsucc_{\widehat{a}}(u'''))$; otherwise set
  $(a,u',v')=(v''',\LOsucc_{\widehat{a}}(v'''),v'')$.

  Each split that does not change any light depths
  (see~Lemma~\ref{lem:light-change-trees-split}) causes a single
  $\LOinsert_{\widehat{u}}(u,v)$ in some heavy path $\widehat{u}$,
  which takes amortized constant time. In this case the new node $v$
  just gets initialized with $\widehat{v}=\widehat{u}$.

  Similarly, each $\contract(v,u)$ that does not change any light depths
  (see~Lemma~\ref{lem:light-change-trees-contract}) causes a single
  $\LOdelete_{\widehat{u}}(v)$ in some heavy path $\widehat{u}$,
  which takes amortized constant time. 

  For any split/contract that changes some light depths, by
  Lemma~\ref{lem:few-light-changes-split-contract} we can afford to
  use constant time per node whose light-depth changes, to update the
  node.  To do this, we first traverse each of the changed subtrees
  from Lemma~\ref{lem:light-change-trees-split}
  or~\ref{lem:light-change-trees-contract} before the split/contract
  in depth-first postorder and uninitialize each node $u$ by using
  $\LOdelete_{\widehat{u}}(u)$, and if $\apex(\widehat{u})=u$ we use
  $\CAdeleteleaf(\widehat{u})$ to delete it from the light
  tree. Finally we remove the reference to $\widehat{u}$ from $u$ to
  mark it as uninitialized.  After a split, if $v$
  has no heavy children we simply (re)initialize $T_v$ using the same
  depth-first preorder procedure as in the initial initialization.  If
  $v$ has a heavy child $c$, we set $\widehat{v}=\widehat{c}$, insert
  $v$ \emph{before} the first element $c$ in the list-ordering
  structure of $\widehat{v}$ (using $\LOinsert_{\widehat{v}}(\nil,v)$)
  and change $\apex(\widehat{v})$ to $v$, then store $\widehat{v}$ in
  $v$ to mark it as initialized, and finally we use the depth-first
  preorder initialization procedure on $T_c$ for each remaining child
  $c$. After a contract, we similarly just use the depth-first
  preorder initialization procedure on $T_c$ for each child
  $c$ that changed its light depth.

  This clearly takes at most constant time per node whose light depth
  changes, so by Lemma~\ref{lem:few-light-changes-split-contract} the
  total time for for any sequence of splits is $O((B+s)\log B)$.
\end{proof}

Note in particular that if every white node has at least two children,
then there are at most $B+s+1 \leq 2B-1$ nodes in total, so $s\leq
B-1$, and thus any sequence of splits/contracts on such trees takes at
most $O(B\log B)$ time.

\begin{theorem}\label{thm:tree-split-contract-ca}
  There is a data structure for dynamic trees that, when initialized
  on a star with a white center and $B$ black leaves, can be
  initialized and handle any sequence of $s$ splits of white nodes and
  (at most $B+s$) contractions in $O(s+B\log B)$ total time, while
  answering intermixed $\parent$, $\firstchild$, $\nextsibling$,
  $\ca$, and $\firstonpath$ queries in worst case constant time.
\end{theorem}
\begin{proof}
  We use the data structure from
  Lemma~\ref{lem:tree-split-contract-parent} to maintain $T$, and the
  data structure from Lemma~\ref{lem:tree-split-contract-ca-slow} to
  maintain a compressed version $\overline{T}$ of $T$ where every node
  with only one child has been contracted with that child.  The tree
  $\overline{T}$ always has at most $2B-1$ nodes, and for any sequence
  of splits/contract in $T$, the corresponding (at most $O(B)$)
  splits/contracts in $\overline{T}$ thus takes at most $O(B\log B)$
  time.  For each node $v$ in $T$, let $\overline{v}$ denote its
  corresponding node in $\overline{T}$. We can easily maintain, in
  amortized constant time per split/contract that creates/removes a
  node with a single child, a list-order structure for each
  $\overline{v}$ with the order representing increasing depth,
  together with the minimum-depth node in the set (denoted
  $\apex(\overline{v})$), as follows:
  \begin{itemize}
  \item
    To do a $\splitnode(u,M)$, if $M=\set{c}$ simply do either a
    $\LOinsert_{\overline{u}}(\parent(u),v)$ and set
    $\overline{v}=\overline{u}$ (if $c=p$), or a
    $\LOinsert_{\overline{c}}(\nil,v)$ and set
    $\overline{v}=\overline{c}$ and $\apex{\overline{v}}=v$
    (otherwise), otherwise $\abs{M}>1$ and we do a
    $\splitnode(\overline{u},\set*{\overline{c}\cond c\in M})$.

  \item
    To do a $\contract(v,u)$: if $d(v)=2$ let
    $c=\LOsucc_{\overline{v}}(v)$, remove $v$ from its list-order
    structure using $\LOdelete_{\overline{v}}$, and if
    $\apex{\overline{v}}=v$ set $\apex(\overline{c})=c$; else if
    $d(u)=2$ remove $v$ from its list-order structure using
    $\LOdelete_{\overline{v}}(v)$; otherwise do a
    $\contract(\overline{v},\overline{u})$ in $M$.
  \end{itemize}
  Thus, the total time for any sequence of $s$ splits and (at most
  $B+s$) contracts is $O(s+B\log B)$ as desired.

  To compute $(a,u',v')=\ca(u,v)$, start by computing
  $(\overline{a},\overline{u}',\overline{v}')$.  Then let
  $u''=\apex(\overline{u}')$ and $u'''=\parent(u'')$ if
  $\overline{a}\neq\overline{u}'$ and $u''=u'''=u$ otherwise, and
  compute $v''$ and $v'''$ symmetrically.  Now if $u'''=v'''$ set
  $(a,u',v')=(u''',u'',v'')$; otherwise if
  $\LOorder_{\overline{a}}(u''',v''')$ set
  $(a,u',v')=(u''',u'',\LOsucc_{\overline{a}}(u'''))$; otherwise set
  $(a,u',v')=(v''',\LOsucc_{\overline{a}}(v'''),v'')$. This takes
  worst case constant time per operation.
\end{proof}

\section{Dynamic tree connectivity with split/contract}
\label{sec:tree-connectivity}

In this section, we will (again) be considering a tree with two kinds of
nodes, called \emph{black} and \emph{white}, and rooted at some node
$r$. For any node $u$, let $N(u)$ denote the set of neighbors of $u$
(including $\parent(u)$ if $u\neq r$), and let $d(u)=\abs{N(u)}$.

We will support the following operations:
\begin{description}
\item[split$(u,M)$:] given a white node $u$ with $d(u)\geq 2$, and a
  subset $M\subset N(u)$ of its neighbors, $1\leq \abs{M}\leq
  \frac{1}{2}d(u)$ such that\footnote{This restriction disallows adjacent white nodes of degree $2$} if $M=\set{c}$ and $c$ is white then
  $d(c)>2$, insert a new white node $v$ as child of $u$.  Then let
  $M'=M$ if $u=r$ or $\parent(u)\not\in M$, and $M'=N(u)\setminus M$
  otherwise.  Finally make each node in $M'$ a child of $v$.

\item[contract$(e)$] given an edge $e=(c,p)$, where at least one of
  $c,p$ is black, make all children of $c$ children of $p$ and turn
  $p$ black.

\item[delete$(e)$] given an edge $e=(c,p)$, where both $c$ and $p$
  are black, delete $e$.

\item[connected$(u,v)$] given nodes $u$ and $v$, return $\true$ if $u$ and $v$ are (still) in the same tree, otherwise $\false$.
\end{description}
The goal is that, when starting on a star with $B$ black leaves and a
white center, any sequence of $s$ splits\footnote{It is possible to
  have $s$ arbitrarily large, because given a white node $u$ of degree
  $d(u)>2$ with a black neighbor $c$ we can repeatedly do
  $v=\splitnode(u,\set{c})$ and $\contract(c,v)$.} and (at most $B+s$)
contracts and deletes takes total $O(s+B\log B)$ time, while
intermixed $\connected$ queries take worst case constant time.

\begin{lemma}\label{lem:tree-delete-properties}
  After starting with a star with a white center and $B$ black leaves,
  and applying any valid sequence of node splits and edge contractions and
  deletions, the following is true:
  \begin{enumerate}
  \item\label{it:black-decreasing} The total number of black nodes is
    at most $B$.
  \item\label{it:black-bounded-tree} Any remaining tree $T$ has at
    most $O(b(T))$ nodes, where $b(T)$ is the number of black nodes in
    $T$.
  \end{enumerate}
\end{lemma}
\begin{proof}
  To see~\ref{it:black-decreasing}, note that the total number of
  black nodes starts at $B$, and that the only operation that can
  change the number of black nodes is $\contract(v,u)$, which can only
  decrease it.

  To see~\ref{it:black-bounded-tree}, note that our definition of
  $\splitnode$, $\contract$, and $\deleteedge$ precludes the
  construction of adjacent white nodes of degree $2$, and of white
  leaves.  Thus the number of white nodes in any such tree $T$ is
  upper bounded by $4b(T)-5\in O(b(T))$.
\end{proof}

\begin{lemma}\label{lem:tree-connectivity}
  There is a data structure for dynamic trees that, when initialized
  on a star with a white center and $B$ black leaves, can be
  initialized and handle any sequence of $s$ node splits and (at most
  $B+s$) edge contractions and deletions in $O(s+B\log B)$ total time,
  while answering intermixed $\parent$, $\firstchild$, $\nextsibling$,
  and $\connected$ queries in worst case constant time.
\end{lemma}
\begin{proof}
  We use the data structure from
  Lemma~\ref{lem:tree-split-contract-parent} as a basis. It can easily
  be extended to track the color of each node, and to handle deletions
  in worst case constant time without affecting the asymptotic time
  for the other operations.  Thus the time spent on splits and
  contractions is still at most $O(s+B\log B)$.

  To handle $\connected(u,v)$ queries, we will assign a
  \emph{component id} $\cid(u)$ to each node.  Then $\connected(u,v)$
  is $\true$ if and only if $\cid(u)=\cid(v)$, which can clearly be
  computed in worst case constant time.

  After a $\splitnode(u,M)$ creates a new node $v$, we simply set
  $\cid(v)=\cid(u)$, and $\contract(v,u)$ doesn't change any $\cid$ at
  all.  This doesn't change the time for processing $\splitnode$ and
  $\contract$.  After a $\deleteedge(v,u)$, we need to create a new
  $\cid$, say $x$. We then traverse the two subtrees $T_1,T_2$ in
  parallel to find the tree $T_i$ minimizing $b(T_i)$.  Since by
  Lemma~\ref{lem:tree-delete-properties} $T_i$ has at most $O(b(T_i))$
  nodes, this search can be done in worst case
  $O(\min\set{b(T_1),b(T_2)})$ time.  Finally, for each $u\in T_i$ we
  set $\cid(u)=x$.  Note that each black node $u$ gets assigned a new
  $\cid(u)$ at most $O(\log B)$ times this way, and the total time
  used on reassignments is proportional to the number of black node
  reassignments, thus the total time used for the reassignments over any
  sequence of deletes is $O(B\log B)$.
\end{proof}

\bibliographystyle{plain}
\bibliography{refs}

\begin{thebibliography}{10}

\bibitem{Aho:1974:DAC:578775}
Alfred~V. Aho, John~E. Hopcroft, and Jeffrey~D. Ullman.
\newblock {\em The Design and Analysis of Computer Algorithms}.
\newblock Addison-Wesley Longman Publishing Co., Inc., Boston, MA, USA, 1st
  edition, 1974.

\bibitem{10.1007/3-540-45749-6_17}
Michael~A. Bender, Richard Cole, Erik~D. Demaine, Martin Farach-Colton, and
  Jack Zito.
\newblock {Two Simplified Algorithms for Maintaining Order in a List}.
\newblock In Rolf M{\"o}hring and Rajeev Raman, editors, {\em Algorithms ---
  ESA 2002}, pages 152--164, Berlin, Heidelberg, 2002. Springer Berlin
  Heidelberg.

\bibitem{ColeHariharan:2005:doi:10.1137/S0097539700370539}
Richard Cole and Ramesh Hariharan.
\newblock Dynamic {LCA} {Q}ueries on {T}rees.
\newblock {\em SIAM Journal on Computing}, 34(4):894--923, 2005.

\bibitem{Dietz:1987:TAM:28395.28434}
Paul Dietz and Daniel Sleator.
\newblock {Two Algorithms for Maintaining Order in a List}.
\newblock In {\em Proceedings of the Nineteenth Annual ACM Symposium on Theory
  of Computing}, STOC '87, pages 365--372, New York, NY, USA, 1987. ACM.

\bibitem{Eppstein:2003}
David Eppstein.
\newblock Dynamic generators of topologically embedded graphs.
\newblock In {\em Proceedings of the Fourteenth Annual ACM-SIAM Symposium on
  Discrete Algorithms}, SODA '03, pages 599--608, Philadelphia, PA, USA, 2003.
  Society for Industrial and Applied Mathematics.

\bibitem{Eppstein97}
David Eppstein, Zvi Galil, Giuseppe~F. Italiano, and Amnon Nissenzweig.
\newblock Sparsification - a technique for speeding up dynamic graph
  algorithms.
\newblock {\em Journal of the ACM}, 44(5):669--696, September 1997.

\bibitem{Eppstein99}
David Eppstein, Zvi Galil, Giuseppe~F. Italiano, and Thomas~H. Spencer.
\newblock Separator-based sparsification ii: Edge and vertex connectivity.
\newblock {\em SIAM Journal on Computing}, 28(1):341--381, February 1999.

\bibitem{EppItaTam92}
David Eppstein, Giuseppe~F. Italiano, Roberto Tamassia, Robert~E. Tarjan,
  Jeffery~R. Westbrook, and Moti Yung.
\newblock {Maintenance of a minimum spanning forest in a dynamic planar graph}.
\newblock {\em Journal of Algorithms}, 13(1):33--54, March 1992.
\newblock Special issue for 1st SODA.

\bibitem{Frederickson85}
Greg~N. Frederickson.
\newblock Data structures for on-line updating of minimum spanning trees, with
  applications.
\newblock {\em {SIAM} Journal on Computing}, 14(4):781--798, 1985.

\bibitem{Frederickson:87}
Greg~N. Frederickson.
\newblock Fast algorithms for shortest paths in planar graphs, with
  applications.
\newblock {\em SIAM Journal on Computing}, 16(6):1004--1022, December 1987.

\bibitem{Frederickson97}
Greg~N. Frederickson.
\newblock Ambivalent data structures for dynamic 2-edge-connectivity and k
  smallest spanning trees.
\newblock {\em SIAM Journal on Computing}, 26(2):484--538, 1997.

\bibitem{Fredman:89}
Michael~L. Fredman and Michael~E. Saks.
\newblock The cell probe complexity of dynamic data structures.
\newblock In {\em Proceedings of the Twenty-first Annual ACM Symposium on
  Theory of Computing}, STOC '89, pages 345--354, New York, NY, USA, 1989. ACM.

\bibitem{Gabow:1990:DSW:320176.320229}
Harold~N. Gabow.
\newblock Data structures for weighted matching and nearest common ancestors
  with linking.
\newblock In {\em Proceedings of the First Annual ACM-SIAM Symposium on
  Discrete Algorithms}, SODA '90, pages 434--443, Philadelphia, PA, USA, 1990.
  Society for Industrial and Applied Mathematics.

\bibitem{Galil:1999}
Zvi Galil, Giuseppe~F. Italiano, and Neil Sarnak.
\newblock Fully dynamic planarity testing with applications.
\newblock {\em Journal of the ACM}, 46(1):28--91, January 1999.

\bibitem{Giammarresi:96}
Dora Giammarresi and Giuseppe~F. Italiano.
\newblock Decremental 2- and 3-connectivity on planar graphs.
\newblock {\em Algorithmica}, 16(3):263--287, 1996.

\bibitem{Goodrich}
Michael~T. Goodrich.
\newblock Planar separators and parallel polygon triangulation.
\newblock {\em Journal of Computer and System Sciences}, 51(3):374 -- 389,
  1995.

\bibitem{Gustedt98}
Jens Gustedt.
\newblock Efficient union-find for planar graphs and other sparse graph
  classes.
\newblock {\em Theoretical Computer Science}, 203(1):123--141, 1998.

\bibitem{Harary69}
Frank Harary.
\newblock {\em Graph Theory}.
\newblock Addison-Wesley Series in Mathematics. Addison Wesley, 1969.

\bibitem{HENZINGER19973}
Monika~R Henzinger, Philip Klein, Satish Rao, and Sairam Subramanian.
\newblock Faster shortest-path algorithms for planar graphs.
\newblock {\em Journal of Computer and System Sciences}, 55(1):3 -- 23, 1997.

\bibitem{Henzinger95}
Monika~R. Henzinger and Han La~Poutr{\'e}.
\newblock Certificates and fast algorithms for biconnectivity in fully-dynamic
  graphs.
\newblock In Paul Spirakis, editor, {\em Algorithms --- ESA '95}, pages
  171--184, Berlin, Heidelberg, 1995. Springer Berlin Heidelberg.

\bibitem{Henzinger97}
Monika~Rauch Henzinger and Valerie King.
\newblock Fully dynamic 2-edge connectivity algorithm in polylogarithmic time
  per operation, 1997.

\bibitem{Henzinger:1999}
Monika~Rauch Henzinger and Valerie King.
\newblock Randomized fully dynamic graph algorithms with polylogarithmic time
  per operation.
\newblock {\em Journal of the {ACM}}, 46(4):502--516, 1999.
\newblock Announced at STOC '95.

\bibitem{HeTh97}
Monika~Rauch Henzinger and Mikkel Thorup.
\newblock Sampling to provide or to bound: With applications to fully dynamic
  graph algorithms.
\newblock {\em Random Struct. Algorithms}, 11(4):369--379, 1997.

\bibitem{Hershberger94}
John Hershberger, Monika Rauch, and Subhash Suri.
\newblock Data structures for two-edge connectivity in planar graphs.
\newblock {\em Theoretical Computer Science}, 130(1):139--161, 1994.

\bibitem{Holm:2001}
Jacob Holm, Kristian de~Lichtenberg, and Mikkel Thorup.
\newblock Poly-logarithmic deterministic fully-dynamic algorithms for
  connectivity, minimum spanning tree, 2-edge, and biconnectivity.
\newblock {\em Journal of the ACM}, 48(4):723--760, July 2001.

\bibitem{Holm18b}
Jacob Holm, Giuseppe~F. Italiano, Adam Karczmarz, Jakub Lacki, and Eva
  Rotenberg.
\newblock {Decremental SPQR-trees for Planar Graphs}.
\newblock In Yossi Azar, Hannah Bast, and Grzegorz Herman, editors, {\em 26th
  Annual European Symposium on Algorithms (ESA 2018)}, volume 112 of {\em
  Leibniz International Proceedings in Informatics (LIPIcs)}, pages
  46:1--46:16, Dagstuhl, Germany, 2018. Schloss Dagstuhl--Leibniz-Zentrum fuer
  Informatik.

\bibitem{Holm17}
Jacob Holm, Giuseppe~F Italiano, Adam Karczmarz, Jakub Lacki, Eva Rotenberg,
  and Piotr Sankowski.
\newblock Contracting a planar graph efficiently.
\newblock In {\em LIPIcs-Leibniz International Proceedings in Informatics},
  volume~87. Schloss Dagstuhl-Leibniz-Zentrum fuer Informatik, 2017.

\bibitem{Holm18a}
Jacob Holm, Eva Rotenberg, and Mikkel Thorup.
\newblock Dynamic bridge-finding in $\widetilde{O}(\log^2 n)$ amortized time.
\newblock In {\em Proceedings of the Twenty-Ninth Annual {ACM-SIAM} Symposium
  on Discrete Algorithms, {SODA} 2018, New Orleans, LA, USA, January 7-10,
  2018}, pages 35--52, 2018.

\bibitem{HuangHKP17}
Shang{-}En Huang, Dawei Huang, Tsvi Kopelowitz, and Seth Pettie.
\newblock Fully dynamic connectivity in \emph{O}(log \emph{n}(log log
  \emph{n})\({}^{\mbox{2}}\)) amortized expected time.
\newblock In {\em Proceedings of the Twenty-Eighth Annual {ACM-SIAM} Symposium
  on Discrete Algorithms, {SODA} 2017, Barcelona, Spain, Hotel Porta Fira,
  January 16-19}, pages 510--520, 2017.

\bibitem{Kapron:2013}
Bruce~M. Kapron, Valerie King, and Ben Mountjoy.
\newblock Dynamic graph connectivity in polylogarithmic worst case time.
\newblock In {\em Proceedings of the Twenty-fourth Annual ACM-SIAM Symposium on
  Discrete Algorithms}, SODA '13, pages 1131--1142, Philadelphia, PA, USA,
  2013. Society for Industrial and Applied Mathematics.

\bibitem{kejlbergrasmussen_et_al:LIPIcs:2016:6395}
Casper Kejlberg-Rasmussen, Tsvi Kopelowitz, Seth Pettie, and Mikkel Thorup.
\newblock {Faster Worst Case Deterministic Dynamic Connectivity}.
\newblock In Piotr Sankowski and Christos Zaroliagis, editors, {\em 24th Annual
  European Symposium on Algorithms (ESA 2016)}, volume~57 of {\em Leibniz
  International Proceedings in Informatics (LIPIcs)}, pages 53:1--53:15,
  Dagstuhl, Germany, 2016. Schloss Dagstuhl--Leibniz-Zentrum fuer Informatik.

\bibitem{Klein:2013}
Philip~N. Klein, Shay Mozes, and Christian Sommer.
\newblock Structured recursive separator decompositions for planar graphs in
  linear time.
\newblock In {\em Proceedings of the Forty-fifth Annual ACM Symposium on Theory
  of Computing}, STOC '13, pages 505--514, New York, NY, USA, 2013. ACM.

\bibitem{LiptonTarjan}
Richard~J. Lipton and Robert~E. Tarjan.
\newblock {A Separator Theorem for Planar Graphs}.
\newblock {\em SIAM Journal on Applied Mathematics}, 36(2):177--189, 1979.

\bibitem{Lacki2011}
Jakub {{\L}\k{a}cki} and Piotr Sankowski.
\newblock Min-cuts and shortest cycles in planar graphs in ${O}(n\log\log{n})$
  time.
\newblock In {\em Algorithms - {ESA} 2011 - 19th Annual European Symposium,
  Saarbr{\"{u}}cken, Germany, September 5-9, 2011. Proceedings}, pages
  155--166, 2011.

\bibitem{Lacki15}
Jakub {{\L}\k{a}cki} and Piotr Sankowski.
\newblock Optimal decremental connectivity in planar graphs.
\newblock In {\em 32nd International Symposium on Theoretical Aspects of
  Computer Science, {STACS} 2015, March 4-7, 2015, Garching, Germany}, pages
  608--621, 2015.

\bibitem{NSW17}
Danupon Nanongkai, Thatchaphol Saranurak, and Christian Wulff{-}Nilsen.
\newblock Dynamic minimum spanning forest with subpolynomial worst-case update
  time.
\newblock In {\em Proceedings of the 58th Annual Symposium on Foundations of
  Computer Science, {FOCS} 2017}, 2017.

\bibitem{patrascu06}
Mihai P\v{a}tra\c{s}cu and Erik~D Demaine.
\newblock Logarithmic lower bounds in the cell-probe model.
\newblock {\em SIAM Journal on Computing}, 35(4):932--963, 2006.

\bibitem{Patrascu:2014:DIS:2706700.2707461}
Mihai P\v{a}tra\c{s}cu and Mikkel Thorup.
\newblock Dynamic integer sets with optimal rank, select, and predecessor
  search.
\newblock In {\em Proceedings of the 2014 IEEE 55th Annual Symposium on
  Foundations of Computer Science}, FOCS '14, pages 166--175, Washington, DC,
  USA, 2014. IEEE Computer Society.

\bibitem{reed:hal-01184376}
Bruce Reed and David~R. Wood.
\newblock {Fast separation in a graph with an excluded minor}.
\newblock In Stefan Felsner, editor, {\em {2005 European Conference on
  Combinatorics, Graph Theory and Applications (EuroComb '05)}}, volume DMTCS
  Proceedings vol. AE, European Conference on Combinatorics, Graph Theory and
  Applications (EuroComb '05) of {\em DMTCS Proceedings}, pages 45--50, Berlin,
  Germany, 2005. {Discrete Mathematics and Theoretical Computer Science}.

\bibitem{Tarjan72}
Robert~E. Tarjan.
\newblock Depth-first search and linear graph algorithms.
\newblock {\em SIAM Journal on Computing}, 1(2):146--160, 1972.

\bibitem{Tarjan:75}
Robert~E. Tarjan.
\newblock Efficiency of a good but not linear set union algorithm.
\newblock {\em Journal of the ACM}, 22(2):215--225, April 1975.

\bibitem{DBLP:journals/dam/TazariM09}
Siamak Tazari and Matthias M{\"{u}}ller{-}Hannemann.
\newblock Shortest paths in linear time on minor-closed graph classes, with an
  application to steiner tree approximation.
\newblock {\em Discrete Applied Mathematics}, 157(4):673--684, 2009.

\bibitem{Thorup97}
Mikkel Thorup.
\newblock Decremental dynamic connectivity.
\newblock In {\em SODA '97}, pages 305--313. SIAM, 1997.

\bibitem{Thorup:2000}
Mikkel Thorup.
\newblock Near-optimal fully-dynamic graph connectivity.
\newblock In {\em Proceedings of the Thirty-second Annual ACM Symposium on
  Theory of Computing}, STOC '00, pages 343--350, New York, NY, USA, 2000. ACM.

\bibitem{Wulff-Nilsen11}
Christian Wulff{-}Nilsen.
\newblock Separator theorems for minor-free and shallow minor-free graphs with
  applications.
\newblock In {\em {IEEE} 52nd Annual Symposium on Foundations of Computer
  Science, {FOCS} 2011, Palm Springs, CA, USA, October 22-25, 2011}, pages
  37--46, 2011.

\bibitem{Wulff-Nilsen16a}
Christian Wulff{-}Nilsen.
\newblock Faster deterministic fully-dynamic graph connectivity.
\newblock In {\em Encyclopedia of Algorithms}, pages 738--741. Springer Berlin
  Heidelberg, Berlin, Heidelberg, 2016.

\end{thebibliography}

\end{document}